\documentclass[a4paper,UKenglish,cleveref,nameinlink,autoref,thm-restate,authorcolumns]{lipics-v2021}

\hideLIPIcs  %
\nolinenumbers

\graphicspath{{./figures}}

\newtheorem{lemma}{Lemma}

\usepackage{cite}
\usepackage{soul}
\usepackage[textsize=footnotesize]{todonotes}
\usepackage{tabularx}
\usepackage[normalem]{ulem}
\usepackage{xspace}
\usepackage{mathtools}
\usepackage{amsthm}
\usepackage{hyperref}
\usepackage[capitalise]{cleveref}
\usepackage[inline]{enumitem}
\usepackage{tcolorbox}


\usepackage{xcolor}

\usepackage{pifont}
\definecolor{PKgreen}{rgb}{0.2,0.627,0.172}
\definecolor{PKred}{rgb}{0.89,0.102,0.109}
\definecolor{MEPKcolor}{rgb}{0.79,0.69,0.84}

\newcommand{\re}{\color{PKred}}

\newcommand{\xmark}{{\re\ding{55}}}

\definecolor{defblue}{rgb}{0.121,0.47,0.705}
\DeclareTextFontCommand{\emph}{\color{defblue}\em}
\hypersetup{colorlinks=false, %
    linkcolor=defblue,
    anchorcolor=defblue,
    citecolor=defblue,
    filecolor=defblue,
    menucolor=defblue,
    urlcolor=defblue,
    bookmarksopen=true,
    bookmarksopenlevel=2,
    bookmarksnumbered=true,
    plainpages=false
    }

\NewDocumentCommand{\lQL}{o}{\ensuremath{\prec\IfNoValueF{#1}{_{#1}}}\xspace}
\NewDocumentCommand{\lQLFull}{o}{\ensuremath{\langle\prec\IfNoValueF{#1}{_{#1}},\sigma\IfNoValueF{#1}{_{#1}}\rangle}\xspace}
\NewDocumentCommand{\Size}{m}{\ensuremath{\vert #1 \vert}\xspace}
\NewDocumentCommand{\Extension}{m}{\ensuremath{\mathtt{#1}}\xspace}

\NewDocumentCommand{\EExtension}{}{\Extension{E}}
\NewDocumentCommand{\OLExtension}{}{\Extension{OL}}
\NewDocumentCommand{\ORExtension}{}{\Extension{OR}}

\newcommand{\defproblem}[4]{
  \begin{tcolorbox}%
    \nolinenumbers\vspace*{-1ex}\hspace*{-2ex}
    \begin{minipage}{0.98\textwidth}
      \begin{tabular}{@{}>{\normalsize}l@{~~}>{\normalsize}p{0.92\textwidth}@{}}
        {\sf\bfseries\color{lipicsGray} Problem:} & #1\\[.1ex]
        {\sf\bfseries\color{lipicsGray} Input:} & #2\\[.1ex]
        {\sf\bfseries\color{lipicsGray} #4:} & #3
      \end{tabular}
    \end{minipage}\vspace*{-1ex}
  \end{tcolorbox}
}
\newcommand{\defdecproblem}[3]{\defproblem{#1}{#2}{#3}{Question}}

\crefname{F}{F}{Fs}
\Crefname{F}{F}{Fs}

\usepackage{apptools,thmtools}
\let\origrestatable=\restatable
\renewcommand{\restatable}[1][]{%
  \origrestatable[%
    \ifstrempty{#1}%
      {$\star$}%
      {\IfAppendix{\hyperref[#1]{$\star$}}{\hyperref[#1*]{$\star$}}}%
  ]%
}
\newbool{restatinbody}
\makeatletter
\patchcmd{\@begintheorem}
  {\thm@swap\swappedhead\thmhead}
  {\ifbool{restatinbody}{\let\thmnote\@gobble}{}\thm@swap\swappedhead\thmhead}
  {}{}
\makeatother
\newcommand{\restateinbody}[1]{\booltrue{restatinbody}#1\boolfalse{restatinbody}}

\title{\texorpdfstring{On the Recognition of Outerplanar Graphs \\ with Queue Number~1}{On the Recognition of Outerplanar Graphs with Queue Number~1}}
\titlerunning{On the Recognition of Outerplanar Graphs with Queue Number~1}

\acknowledgements{This research was initiated at the Bertinoro Workshop on Graph Drawing 2026.}

\author{Michael A. Bekos}{University of Ioannina, Greece \and \url{https://myweb.uoi.gr/bekos/} }{bekos@uoi.gr}{https://orcid.org/0000-0002-3414-7444}
{Supported by HFRI Grant No: 26320.}

\author{Thomas Depian}{TU Wien, Vienna, Austria}{tdepian@ac.tuwien.ac.at}{https://orcid.org/0009-0003-7498-6271}{Supported by the Vienna Science and Technology Fund (WWTF) [10.47379/ICT22029].}

\author{Stefan Felsner}{Technische Universität Berlin, Germany}{felsner@math.tu-berlin.de}{}{}

\author{Michael Kaufmann}{%
University of T{\"u}bingen, Germany}{michael.kaufmann@uni-tuebingen.de}{}{}

\author{Philipp Kindermann}{Universit\"at Trier, Germany \and \url{https://algo.uni-trier.de/~kindermann} }{kindermann@uni-trier.de}{https://orcid.org/0000-0001-5764-7719}{}

\author{Fabrizio Montecchiani}{University of Perugia, Italy}{fabrizio.montecchiani@unipg.it}{https://orcid.org/0000-0002-0543-8912}{}

\author{Maria Eleni Pavlidi}{University of Ioannina, Greece}{m.e.pavlidi@uoi.gr}{https://orcid.org/0009-0009-4500-0112}{Supported by HFRI Grant No: 26320.}

\author{Alexandra Weinberger}{FernUniversit\"at in Hagen, Germany \and \url{https://www.fernuni-hagen.de/ti/team/alexandra.weinberger.shtml}}{alexandra.weinberger@fernuni-hagen.de}{https://orcid.org/0000-0001-8553-6661}{}

\author{Alexander Wolff}{Universität Würzburg, Germany \and \url{https://www.informatik.uni-wuerzburg.de/en/algo/team/wolff-alexander/}}{}{https://orcid.org/0000-0001-5872-718X}{}

\author{Johannes Zink}{Technische Universität München, Germany \and \url{https://www.cs.cit.tum.de/algo/staff/johannes-zink/}}{johannes.zink@tum.de}{https://orcid.org/0000-0002-7398-718X}{}

\authorrunning{Bekos et al.} %


\definecolor{dark blue}{rgb}{0.121,0.47,0.705}
\let\emph\relax\DeclareTextFontCommand{\emph}{\color{dark blue}\em}

\newtheorem{open}{Open Problem}
\theoremstyle{claimstyle}

\begin{document}

\maketitle

\vspace{-4pt}

\begin{abstract}
A \emph{linear layout} of a graph is defined as a total order of the vertices and a partition of the edges to pages.
In a \emph{stack (queue) layout}, no two edges on the same page may cross (nest).
The \emph{stack (queue) number} of a graph is the minimum number of pages required in a stack (queue) layout.
This paper focuses on characterizing and recognizing graphs that have both stack number~1 and queue number~1.
It is known that the graphs with stack number~$1$ are exactly the outerplanar graphs.
We show that (i)~deciding whether a given outerplanar graph has queue
number~$1$ is NP-hard; (ii)~deciding whether a given {\em maximal}
outerplanar graph has queue number~$1$ can be done in linear time.
Moreover, we investigate the interplay between outerpaths with queue
number~$1$ and their maximum vertex degree.
\end{abstract}

\vspace{-4pt}

\tableofcontents

\section{Introduction}

A linear layout of a graph is defined by a total order of its vertices together with a partition of its edges into pages such that edges belonging to the same page satisfy a specific property; see, e.g., \cite{2004layouts} for a survey.
In this work, we focus on a well-studied variant of linear layouts called \emph{queue layouts}~\cite{DBLP:journals/siamcomp/HeathR92}, in which each page is required to form a queue, that is, no two of its edges are nested with respect to the vertex order.
Given a graph, it is an interesting problem to determine its \emph{queue number}, defined as the minimum number of pages required in a queue layout of the graph.
Note that the counterpart of a queue layout is a \emph{stack layout}~\cite{DBLP:journals/jct/BernhartK79}, in which each page is required to form a stack, which does not allow any two of its edges to cross.
Analogously, the \emph{stack number} (also called \emph{page number} or \emph{book thickness}) of a graph is the minimum number of pages in a stack layout.

Back in 1992, Heath and Rosenberg~\cite{DBLP:journals/siamcomp/HeathR92} characterized the graphs admitting single-queue layouts as the arched leveled-planar graphs (which we define in \cref{sec:preliminaries}) and showed that
recognizing these graphs is NP-hard.
Interestingly, this complexity-theoretic barrier does not carry over to stack layouts, since the graphs admitting single-stack layouts are precisely the  outerplanar graphs~\cite{DBLP:journals/jct/BernhartK79}, which can be recognized in linear time~\cite{DBLP:conf/wg/Wiegers86}.
As a matter of fact, outerplanar graphs have several useful structural properties that allow for efficient algorithms for many problems that are NP-hard in general, e.g., they are characterized by forbidden minors ($K_{2,3}$ and $K_4$)~\cite{79stack}, have degeneracy and treewidth at most~2, and chromatic number at most $3$~\cite{LickWhite1970}.
Since all outerplanar graphs admit 2-queue layouts but not all outerplanar graphs admit single-queue layouts~\cite{DBLP:journals/siamdm/HeathLR92}, a natural question is whether the problem of testing whether a graph has queue number~1 remains computationally hard when restricted to outerplanar graphs.
A negative answer would establish an intriguing parallel between stack and queue layouts in the outerplanar setting, which to the best of out knowledge has not been investigated yet.
It would also imply that the graphs that have both stack number~$1$ and queue number~$1$ (but not necessarily on the same underlying vertex order) can be efficiently recognized.
The relevance of the class of outerplanar graphs with queue number~$1$ extends beyond linear layouts: in the context of so-called \emph{simple drawings}, this class of graphs was shown to consist precisely of the graphs guaranteed to occur as plane subdrawings in all simple drawings of sufficiently large complete graphs~\cite{wz-wips-GD26}.

Towards an answer to the question above, the characterization by Heath and Rosenberg~\cite{DBLP:journals/siamcomp/HeathR92} implies that determining whether an outerplanar graph has queue number~$1$ is equivalent to determining whether it is arched leveled-planar.
In this direction, Bannister, Devanny, Dujmović, Eppstein, and Wood~\cite{DBLP:journals/algorithmica/BannisterDDEW19} proved that every bipartite outerplanar graph is leveled outerplanar (and thus arched leveled-planar), and therefore admits a 1-queue layout.
Beyond the bipartite case, however, the problem turns out to be widely open.

\subparagraph{Our Contribution.}

First, we settle the complexity of recognizing outerplanar graphs with queue number~1. By a reduction from \textsc{3-Partition}, we show that deciding whether a given outerplanar graph has queue number~1 is NP-hard (see \cref{sec:hard}). This result strengthens the corresponding NP-hardness result by Heath and Rosenberg~\cite{DBLP:journals/siamcomp/HeathR92}, who established hardness for general planar graphs.

\begin{restatable}[thm:nphardness]{theorem}{nphardness}
    \label{thm:nphardness}
    Determining whether an outerplanar graph has queue number $1$ is NP-hard, even when restricted to outerplanar graphs of pathwidth at most $5$.
\end{restatable}

The fact that the graphs that we construct in our reduction have pathwidth at most~5 has immediate implications for the parameterized complexity of computing the queue number with respect to the usual structural parameters.
More concretely, the problem is known to admit a fixed-parameter algorithm with respect to the vertex integrity of the graph~\cite{DFGS.LLR.2025}. 
Only for the special case of queue number~$1$, a fixed-parameter algorithm with respect to treedepth has been established~\cite{BGMN.PAQ.2022}.
Generalizing either of the two algorithms to less restrictive parameters, in particular treewidth, has been repeatedly posed as an open problem~\cite{GMNZ.PCG.2021,BSK.Bop.2025}.
With our result, we provide a negative answer to this question:

\begin{restatable}[cor:parahardness-treewidth]{corollary}{paranphardness}
    \label{cor:parahardness-treewidth}
    Determining the queue number of a graph is paraNP-hard
    with respect to pathwidth plus  queue number.
\end{restatable}

For maximal outerplanar graphs, we have a complete characterization.  
We first show that if such a graph has queue number~1, then it must be an \emph{outerpath}, that is, an outerplanar graph whose weak dual is a path. 

\begin{restatable}[thm:maxouterplanar]{theorem}{maxouterplanar}
\label{thm:maxouterplanar}
Every maximal outerplanar graph with queue number~$1$ is an outerpath. 
\end{restatable}

Then we give a linear-time algorithm that recognizes maximal outerpaths with queue number~1; see \cref{sec:linear-time}.

\begin{restatable}[thm:lineartime]{theorem}{lineartime}
  \label{thm:lineartime}
  Given a maximal outerpath~$G$, one can decide in linear time whether
  $G$ has queue number~$1$ and, if yes, construct a $1$-page queue
  layout within the same time bound.
\end{restatable}

Finally, we investigate the interplay between outerpaths with queue
number~$1$ and their maximum vertex degree.  We have the following
existential results; see \cref{sec:outerpaths}.

\begin{restatable}[prop:max-degree]{proposition}{maxDegreeProposition}
	\label{prop:max-degree}
        \strut
	\begin{enumerate}
		\item Every maximal outerpath with queue number~$1$ has vertex degree at most~$6$.
		\item Every maximal outerpath with vertex degree at most $4$ has queue number~$1$.
		\item Every outerpath with vertex degree at most~$3$ has queue number~$1$.
		\item There are outerpaths with vertex degree at most $4$ whose queue number is~$2$.
	\end{enumerate}
\end{restatable}

\section{Preliminaries}
\label{sec:preliminaries}

For an integer $p \geq 1$, we use \emph{$[p]$} as a shorthand for the set $\{1, 2, \dots, p\}$.
For a simple graph~$G$, we let \emph{$V(G)$} denote the vertex set of~$G$ and we let \emph{$E(G)$} $\subseteq {V(G) \choose 2}$ denote the edge set of~$G$.
Without loss of generality, we assume $G$ to be connected and nonempty.
For a vertex~$v$ of~$G$, we let \emph{$N_G(v)$} $= \{u \in V(G) \mid uv \in E(G)\}$ denote the \emph{neighborhood} of~$v$ in~$G$.
We let $\deg_G(v) = |N_G(v)|$ denote the degree of~$v$ in~$G$, and we let \emph{$\Delta(G)$} $= \max_{v \in V(G)} \deg_G(v)$ denote the \emph{maximum degree} of~$G$. 

\subparagraph*{Outerplanar Graphs and Outerpaths.}
A \emph{drawing} $\Gamma_G$ of $G$ maps every vertex $v$ of $G$ to a point $\Gamma(v) \in \mathbb{R}^2$ and every edge $uv$ of $G$ to a curve between $\Gamma(u)$ and $\Gamma(v)$.
We call $\Gamma_G$ \emph{planar} if every two edges meet only at common endpoints.
A planar drawing subdivides the plane into regions, called \emph{faces}.
The \emph{outer face} is the single face of $\Gamma_G$ that is not bounded, the remaining faces are the \emph{inner faces}.
The \emph{weak dual} of a planar drawing $\Gamma_G$ contains one vertex for every inner face of $\Gamma_G$.
Two vertices of the weak dual are adjacent if and only if their respective faces are separated by an edge.
A graph $G$ is called \emph{outerplanar} if it admits a planar drawing $\Gamma_G$ where every vertex lies on the outer face.
An outerplanar graph~$G$ is an \emph{outerpath} if the weak dual of $\Gamma_G$ is a path.
Note that outerpaths are biconnected. 
We say that a graph $G$ is \emph{maximal} with respect to a property (e.g., being outerplanar) if, for every non-edge $uv$ of $G$, the graph $G+uv$ does not have the property.

\subparagraph{(Arched) Leveled-Planar Graphs.} A planar graph is \emph{leveled-planar} if its vertices can be partitioned into levels $V_1,\ldots,V_k$ so that every edge connects only consecutive levels and, for each $\ell \in \{1,\dots,k\}$, the vertices in $V_\ell$ can be placed on the line $y=\ell$ such that the resulting straight-line drawing is planar.
A planar graph $G$ is \emph{arched leveled-planar} if it contains a set~$E'$ of edges, so-called \emph{arches}, so that $G'=G - E'$ admits a leveled-planar drawing~$\Gamma'$ that can be extended to a planar drawing of~$G$.
The arches must connect vertices that lie on the same level in~$\Gamma'$ and, in every level, all arches must be incident to the leftmost vertex.
Additionally, arches must fulfill the following technical requirement.
For~$\ell \in \{1,\dots,k-1\}$, let $V_\ell'$ be the subset of $V_\ell$ of those vertices that have neighbors in~$V_{\ell+1}$.
Then the arches incident to vertices in~$V_\ell$ cannot have their right endpoint to the left of the rightmost vertex in~$V_\ell'$.

\subparagraph*{Stack and Queue Layouts.}
A \emph{linear layout} of a graph $G$ consists of a total order ($\prec$) of $V(G)$, called the \emph{spine} and a partition of the edges $E(G)$ into \emph{pages}. For two vertices $u,v \in V(G)$, we say that $u$ is \emph{left} of $v$ if $u \prec v$ and \emph{right} of $v$ if $v \prec u$.
An \emph{interval} of $\prec$ is a (possibly empty) set of vertices which forms a contiguous block in $\prec$.
Let $tu, vw \in E(G)$ be two edges of $G$ with $t \prec u$, $v \prec w$, and $t \prec v$.
The edges \emph{cross} under $\prec$ if $t \prec v \prec u \prec w$ and they \emph{nest} under $\prec$ if $t \prec v \prec w \prec u$. 
To this end, a \emph{stack layout}~\cite{DBLP:journals/jct/BernhartK79}  is a linear layout in which no two edges assigned to the same page cross, while a \emph{queue layout}~\cite{DBLP:journals/siamcomp/HeathR92} is a linear layout in which no two edges assigned to the same page nest.
The minimum number of pages required in a stack (queue)  layout of G is the stack (queue) number of G. In particular, a graph has stack number one if and only if it admits a vertex order in which no two edges cross, also called a \emph{1-stack layout}, and it has queue number one if and only if it admits a vertex order in which no two edges nest, also called a \emph{1-queue layout}. Furthermore, graphs with stack number one are precisely the outerplanar graphs, while graphs with queue number one correspond to graphs admitting an arched leveled-planar representation~\cite{79stack,DBLP:journals/siamcomp/HeathR92}.
\section{Complexity of the Recognition Problem for General Outerplanar Graphs}
\label{sec:hard}

In this section, we establish NP-hardness for deciding whether an outerplanar graph admits a 1-queue layout.
We reduce from \textsc{3-Partition}, which is defined as follows.

\defdecproblem{\sc 3-Partition}{A multiset $S$ of integers such that $m=|S|/3$ and $B=\sum S/m$ are integers and, for every $s \in S$, $B/4 < s < B/2$.}{Does there exist a partition of~$S$ into $m$ multisets
such that the sum of the numbers in each multiset equals~$B$?}
The task is to decide whether there is a partition of~$S$ into $m$ sets
such that the sum of the numbers in each set equals~$B$.
Since \textsc{3-Partition} is {\em strongly} NP-hard,
we can assume that $B$ is polynomial in~$m$.

Our NP-hardness proof
is a polynomial-time reduction from
\textsc{3-Partition} to
the problem of recognizing
an arched leveled-planar graph.
Recall that a graph admits
an arched leveled-planar drawing if and only if
it admits a 1-queue layout~\cite{DBLP:journals/siamcomp/HeathR92}.
Given an arched leveled-planar drawing~$\Gamma$ and a vertex~$v$,
we let~$y_{\Gamma}(v)$ denote the level of~$v$ in~$\Gamma$.
Essentially, we will use the following observation.

\begin{observation}
	\label{obs:k-odd-length-cycles-k-arches}
    A graph with $k$ edge-disjoint odd-length cycles requires $k$~arches in an arched leveled-planar drawing.
\end{observation}

We will attach our construction modeling the \textsc{3-Partition}
instance to an \emph{anchor gadget}, which is a maximal outerpath that
consists of five triangles; see \cref{fig:nph:anchor-gadget-a}.  The
central triangle of the path consists of three vertices of degree~4,
which we call~$a$, $b$, and~$e$.  The edge~$ab$ lies on the outer face
of the unique outerplanar embedding
of the gadget.  We now show that the gadget behaves in a specific way
that allows us to attach two copies of our construction to~$a$
and~$b$.

\begin{figure}[tb]
    \begin{subfigure}[t]{.12\textwidth}
        \centering
        \includegraphics[page=9]{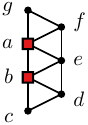}
        \subcaption{}
        \label{fig:nph:anchor-gadget-a}
    \end{subfigure}
    \hfill
    \begin{subfigure}[t]{.21\textwidth}
        \centering
        \includegraphics[page=13]{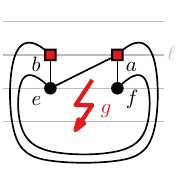}
        \subcaption{Case 1a}
        \label{fig:nph:anchor-gadget-b}
    \end{subfigure}
    \hfill
    \begin{subfigure}[t]{.21\textwidth}
        \centering
        \includegraphics[page=16]{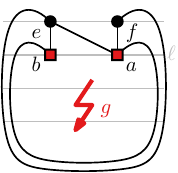}
        \subcaption{Case 1b}
        \label{fig:nph:anchor-gadget-c}
    \end{subfigure}
    \hfill
    \begin{subfigure}[t]{.21\textwidth}
        \centering
        \includegraphics[page=21]{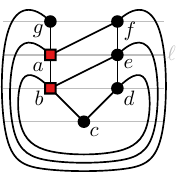}
        \subcaption{Case 2}
        \label{fig:nph:anchor-gadget-d}
    \end{subfigure}
    \hfill
    \begin{subfigure}[t]{.21\textwidth}
        \centering
        \includegraphics[page=26]{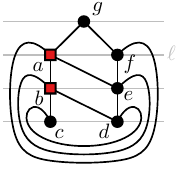}
        \subcaption{Case 3a}
        \label{fig:nph:anchor-gadget-e}
    \end{subfigure}
    
    \bigskip
    
    \begin{subfigure}[t]{.24\textwidth}
        \centering
        \includegraphics[page=27]{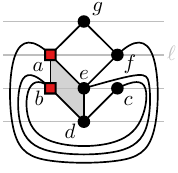}
        \subcaption{Case 3b}
        \label{fig:nph:anchor-gadget-f}
    \end{subfigure}
    \hfill
    \begin{subfigure}[t]{.24\textwidth}
        \centering
        \includegraphics[page=33]{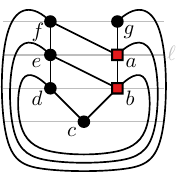}
        \subcaption{Case 4a}
        \label{fig:nph:anchor-gadget-g}
    \end{subfigure}
    \hfill
    \begin{subfigure}[t]{.24\textwidth}
        \centering
        \includegraphics[page=34]{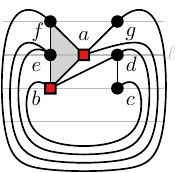}
        \subcaption{Case 4b}
        \label{fig:nph:anchor-gadget-h}
    \end{subfigure}
    \hfill
    \begin{subfigure}[t]{.24\textwidth}
        \centering
        \includegraphics[page=39]{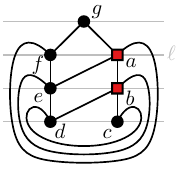}
        \subcaption{Case 5}
        \label{fig:nph:anchor-gadget-i}
    \end{subfigure}
    
    \caption{Anchor gadget of size~$k$.
        (a) shows the graph, (b) and (c) illustrate that $a$ and $b$
        cannot lie on the same level, while (d)--(i) show
        (up to symmetry w.r.t.\ $a$ and~$b$) the only six realizations of the anchor gadget
        in an arched leveled-planar drawing.}
    \label{fig:nph:anchor-gadget}
\end{figure}

\begin{lemma}
	\label{lem:anchor-gadget}
    Let $\Gamma$ be an arched leveled-planar drawing of the anchor gadget
    that has been extended by two vertex-disjoint cycles $C_a$ and $C_b$
    such that $C_a$ contains~$a$ and $C_b$ contains~$b$.
    Then, there is a vertex $v \in \{a, b\}$ such that,
    for every vertex $u \in V(C_v) \setminus \{v\}$,
    $y_{\Gamma}(v) < y_{\Gamma}(u)$.
    In other words, the cycle $C_v$ has its unique bottom point
    at~$v$ and is otherwise drawn above~$v$.
\end{lemma}

\begin{proof}
	We first analyze all realizations of the anchor gadget (without $C_a$ and $C_b$)
	as arched leveled-planar drawings.
	To this end, we go through several cases and subcases.
	To realize the triangle $\triangle abe$, two of its vertices need to lie on the same level.
	Let~$\ell$ be the level of~$a$.
	
	\subparagraph*{Case 1: $b$ lies on level $\ell$, too.}
	Due to the symmetry of the anchor gadget, we may assume
    that $b$ lies to the left of~$a$. %
	We further distinguish where $e$ lies.
	
	\smallskip \noindent \textit{Case 1a: $e$ lies on level $\ell - 1$.}
	See \cref{fig:nph:anchor-gadget-b}.
	To realize $\triangle aef$, $f$ cannot lie on level~$\ell$
    (because the edge~$af$ would arch but $b$ is the apex of the star of arching edges on
    level~$\ell$, and hence leftmost on its level).  Thus, $f$
    lies on level $\ell - 1$, and the edge $ef$ arches.
	Note that $f$ lies to the right of $e$ since
	otherwise the edge $af$ crosses $be$.
    Since $g$ is adjacent to both~$a$ and~$f$,
	$g$ needs to lie on the same level as~$a$ or~$f$.
	Then, either the edge~$ag$ or the edge~$fg$ is arching;
	neither of those, however, can arch since
	$b$ and $e$ are the leftmost vertices
	of the corresponding levels due to their arches.
	So, $g$ cannot be added, and this subcase has no realization.
    
	\smallskip \noindent \textit{Case 1b: $e$ lies on level $\ell + 1$.}
	See \cref{fig:nph:anchor-gadget-c}.
	By an argument symmetric to that in Case~1a,
    $f$ lies on level $\ell + 1$ to the right of~$e$.
	Again, we cannot add $g$ together with an arching edge $ag$ or $fg$
	since $e$ and $b$ are the leftmost vertices of the corresponding levels.
	So, this subcase has no realization.
	
	Since Case~1 has no realization, $a$ and $b$
	must lie on distinct neighboring levels.
	Due to the symmetry of the anchor gadget,
	we can switch the roles of~$a$ and~$b$.
	This allows us to assume henceforth that $b$ lies on level $\ell - 1$.
	
	\subparagraph*{Case 2: $e$ lies on level~$\ell$ to the right of~$a$.}
	See \cref{fig:nph:anchor-gadget-d}.
	To complete $\triangle aef$, $f$ lies on a neighboring level;
	more precisely, $f$ lies on level~$\ell + 1$ since otherwise
	$af$ would cross $be$, or $ef$ would cross $ab$.
	To complete $\triangle afg$, $g$ lies on level $\ell + 1$ to the left of $f$;
	$g$ does not lie on level $\ell$ as otherwise $e$ and $g$
	would be right endpoints of arches on $\ell$ that both have an edge going up.
	On level $\ell + 1$, $g$ cannot lie to the right of $f$
	as this would imply a crossing between $ag$ and $ef$.
	As every triangle in an arched leveled-planar drawing,
    $\triangle bde$ must contain an arch.
	Since $e$ is not leftmost on its level,
    $e$ cannot be involved in a second arch.
	Hence $d$ lies on level $\ell - 1$ and $bd$ arches.
	Moreover, $d$ lies to the right of $b$
	since otherwise $ab$ and $de$ would cross.
	Finally, $c$ lies on level~$\ell - 2$;
	$c$ must lie on a level neighboring the level of $b$ and $d$
	and it cannot lie above those two as this would
	imply a crossing with an edge incident to~$e$.
	This finishes the description of the only realization in this case.

	\subparagraph*{Case 3: $e$ lies on level~$\ell - 1$ to the right of~$b$.}
	See \cref{fig:nph:anchor-gadget-e,fig:nph:anchor-gadget-f}.
	Since $be$ arches, $b$ is the leftmost vertex of level~$\ell - 1$.
	To complete $\triangle aef$, one of the edges incident to~$f$ arches;
	this edge is $af$ since neither $e$ nor $f$
	can be the leftmost vertex on level~$\ell - 1$.
	Note that $f$ lies on level~$\ell$ to the right of~$a$
	because otherwise $ab$ and $ef$ would cross.
	To complete $\triangle afg$, $g$ lies on level $\ell + 1$;
	$g$ cannot lie on level~$\ell-1$ as this would imply a crossing
	between an edge incident to~$e$ and one incident to~$g$.
	To complete $\triangle bde$, $d$ lies on level $\ell - 2$;
	$d$~cannot lie on level~$\ell$ as this would imply a crossing
	between an edge incident to~$a$ and one incident to~$d$.
	To complete $\triangle bcd$, either $bc$ or $cd$ arches.
	
	\smallskip \noindent \textit{Case 3a: $c$ lies on level $\ell - 2$.}
	See \cref{fig:nph:anchor-gadget-e}.
	In this subcase, $c$ lies to the left of $d$
	as otherwise $bc$ and $de$ would cross.
	This finishes the description of the first of
        two realizations in this case.
	
	\smallskip \noindent \textit{Case 3b: $c$ lies on level $\ell - 1$.}
	See \cref{fig:nph:anchor-gadget-f}.
	In this subcase, $c$ lies to the right of~$e$
	because both are right endpoints of arches
    and $c$ must not lie inside the face~$abde$ (shaded in the figure)
    as condition~(iv) for arching would be violated there.
	This finishes the description of the second of
    two realizations in this case.

	\subparagraph*{Case 4: $e$ lies on level~$\ell$ to the left of~$a$.}
	See \cref{fig:nph:anchor-gadget-g,fig:nph:anchor-gadget-h}.
	To complete $\triangle aef$, $f$~lies on a level neighboring
    level~$\ell$.  Indeed, $f$ must lie on level~$\ell + 1$ since
    otherwise an edge incident to~$f$ would cross one incident to~$b$.
	To complete $\triangle afg$, $g$ lies on level $\ell + 1$ to the right of $f$;
	$g$ does not lie on level $\ell$ as $e$ is the leftmost vertex of level~$\ell$
	due to the arching edge~$ae$, which disallows $ag$ to arch.
	Also, $g$ does not lie on level $\ell + 1$ to the left of $f$
	as this would imply a crossing between $ag$ and $ef$.
	To complete $\triangle bde$, either $bd$ or $be$ arches.
	
	\smallskip \noindent \textit{Case 4a: $d$ lies on level $\ell - 1$.}
	See \cref{fig:nph:anchor-gadget-g}.
	In this subcase, $d$ lies to the left of $b$
	as otherwise $ab$ and $de$ would cross.
	Finally, $c$ lies on level~$\ell - 2$;
	$c$ must lie on a level neighboring the level of $b$ and $d$
	and it cannot lie above those two as this would imply a
    crossing between an edge incident to~$d$ and one incident to~$e$.
	This finishes the description of the first of
    two realizations in this case.
	
	\smallskip \noindent \textit{Case 4b: $d$ lies on level $\ell$.}
	See \cref{fig:nph:anchor-gadget-h}.
	In this subcase, $d$ lies to the right of~$a$
	because both are right endpoints of arches
    and $d$ must not lie inside the face~$afeb$ (shaded in the figure).
	Finally, $c$ lies on level~$\ell - 1$ to the right of~$b$;
	$c$ cannot lie on level $\ell$ as the edge $cd$ would need to arch
	but $e$ is already the leftmost vertex of level~$\ell$.
	Also, $c$ lies to the right of~$b$ as otherwise
	$cd$ and $ab$ would cross.
	This finishes the description of the second of
    two realizations in this case.

	\subparagraph*{Case 5: $e$ lies on level~$\ell - 1$ to the left  of~$b$.}
	See \cref{fig:nph:anchor-gadget-i}.
	Since $be$ arches, $e$ is the leftmost vertex of level~$\ell - 1$.
	To complete $\triangle aef$,
	$f$ lies on level $\ell$ or $\ell - 1$;
	$f$ does not lie on level $\ell - 1$
	since then $b$ and $f$ would be right endpoints of arches
	on the same level and both have edges going up.
	Note that $f$ lies on level~$\ell$ to the left of~$a$
	because otherwise $ab$ and $ef$ would cross.
	To complete $\triangle afg$, $g$ lies on level $\ell + 1$;
	$g$ cannot lie on level~$\ell-1$ as this would imply a crossing
	between an edge incident to~$e$ and one incident to~$g$.
	To complete $\triangle bde$, $d$ lies on level $\ell - 2$;
	$d$ cannot lie on level~$\ell$ as this would imply a crossing
	between edges incident to~$a$ and~$d$.
	To complete $\triangle bcd$, either $bc$ or $cd$ arches;
	$bc$ cannot arch since $e$ is already the leftmost vertex of level~$\ell-1$.
	Hence, $c$ lies on level~$\ell - 2$ and $cd$ arches.
	Note that $c$ lies to the right of~$d$
	because otherwise $bc$ and $de$ would cross.
	This finishes the description of the only realization in this case.
	
	\subparagraph*{Wrap-up.}
	Observe that this case analysis is complete
	and there are~-- up to symmetry~-- precisely the six
    distinct realizations depicted in \cref{fig:nph:anchor-gadget-d,fig:nph:anchor-gadget-e,fig:nph:anchor-gadget-f,fig:nph:anchor-gadget-g,fig:nph:anchor-gadget-h,fig:nph:anchor-gadget-i}.
	To complete the proof, we attach, in each of the realizations,
	the cycles $C_a$ and $C_b$ to $a$ and $b$, respectively.
	We have broken symmetries in such a way that $a$ always takes
    the role of~$v$.  It remains to show that no neighbor of~$a$
    in~$C_a$ lies on level~$\ell$ or level~$\ell - 1$.
    Note that
    it is impossible that both neighbors of~$a$ in~$C_a$ lie on
    level~$\ell + 1$ and at the same time another vertex of~$C_a$
    distinct from~$a$ lies on level~$\ell$.
    For an illustration, see \cref{fig:nph:no-going-around}.
    If this was the case,
    this would violate the arching conditions of an edge~$yz$ of
    the anchor gadget arching on level~$\ell$ because $y$ would
    not be the leftmost vertex or $z$ would have a vertex to its
    right with an edge going up.
    Also, entering level~$\ell$ between $y$ and $z$
    is blocked in all realizations
    by other vertices and edges of the anchor gadget.

	\begin{figure}[tb]
		\begin{subfigure}[t]{.45\textwidth}
			\centering
			\includegraphics[page=40]{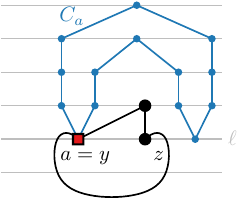}
			\subcaption{Arching condition is violated at~$z$
			(a vertex to the right of~$z$ has an
			edge to the level above).}
			\label{fig:nph:no-going-around-a}
		\end{subfigure}
		\hfill
		\begin{subfigure}[t]{.45\textwidth}
			\centering
			\includegraphics[page=41]{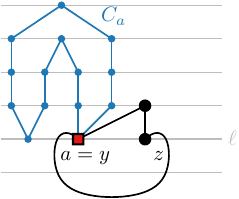}
			\subcaption{Arching condition is violated at~$y$
			($y$ is not the leftmost vertex).}
			\label{fig:nph:no-going-around-b}
		\end{subfigure}
		
		\caption{The cycle~$C_a$ cannot have a vertex on level~$\ell$
			if the two neighbors of $a$ on~$C_a$ are on level~$\ell + 1$.
			Note that here $a = y$ is the leftmost vertex on level~$\ell$.
			The situation is similar if $a = z$ is not the leftmost vertex on level~$\ell$.}
		\label{fig:nph:no-going-around}
	\end{figure}

    Let $x$ be one of the two neighbors of~$a$ in~$C_a$.
	First suppose that the edge $ax$ of~$C_a$ arches.
        
    In Cases 4a, 4b, and~5
    (\cref{fig:nph:anchor-gadget-g,fig:nph:anchor-gadget-h,fig:nph:anchor-gadget-i}),
    this is excluded since either~$e$ or~$f$ is the leftmost
    vertex on level~$\ell$.

    In Case~2 (\cref{fig:nph:anchor-gadget-d}), $x$ would
    have to lie to the right of~$e$.  If the successor of~$x$
    along~$C_a$ was on level~$\ell + 1$, this would violate the
    arching conditions of~$ae$.  If the successor of~$x$
    along~$C_a$ was on level~$\ell - 1$, this would violate the
    arching conditions of~$bd$ or this would imply a crossing.
        
	In Case~3 (\cref{fig:nph:anchor-gadget-e,fig:nph:anchor-gadget-f}),
    $x$ would have to lie to the right of~$f$.
	If the successor of~$x$ along~$C_a$ was on level~$\ell + 1$,
	this would violate the arching conditions of~$af$.
	If the successor of~$x$ along~$C_a$ was on level~$\ell - 1$,
	this would violate the arching conditions of~$be$
    or this would imply a crossing.
	
	Finally suppose that a neighbor $x$ of~$a$ in~$C_a$ lies on level~$\ell - 1$.
	Since in all realizations there is at least one arch on level~$\ell - 1$,
	$x$ cannot be the leftmost vertex on level~$\ell - 1$.
	Therefore, in Case~2 (\cref{fig:nph:anchor-gadget-d}), $ax$ and $be$ would cross,
	and in Case~4b (\cref{fig:nph:anchor-gadget-h}), $ax$ and $bd$ would cross.
	Now suppose that~$x'$ is the successor of~$x$ along~$C_a$ and $x' \ne a$.
	In Case~3 (\cref{fig:nph:anchor-gadget-e,fig:nph:anchor-gadget-f} and Case~5 (\cref{fig:nph:anchor-gadget-i}),
	$x$ could be placed between~$b$ and~$e$,
	but then $xx'$ cannot arch and would therefore cross an edge of the anchor gadget.
	In Case~4a (\cref{fig:nph:anchor-gadget-g}),
	placing $x$ between $d$ and~$b$ implies a crossing between $ax$ and $be$,
	while placing $x$ to the right of $b$ violates the arching conditions of $bd$.
        
	This completes the proof.
\end{proof}

We now introduce a (sub)graph~$T_k$ called \emph{tower gadget of size $k$}.
We use this name because, under specific circumstances (that we get from the anchor gadget),
it requires a height depending on~$k$ in an (arched) leveled-planar drawing.
For an illustration, see \cref{fig:nph:tower-gadget-a}.
The graph $T_k$ is a sequence of $k$ 4-cycles (called \emph{main} 4-cycles)
where each main 4-cycle shares one edge with its predecessor and the opposite edge with its successor.
Moreover, each main 4-cycle has another \emph{additional} 4-cycle attached on
one of its unshared edges; the additional 4-cycles
are all attached on the same side (say on the left if we traverse the main 4-cycles upward).
We call the degree-3 vertex shared between the first main 4-cycle but not the second main 4-cycle~$v$.

\begin{figure}[t]
	\centering
	\begin{subfigure}[t]{.16\textwidth}
		\centering
		\includegraphics[page=3]{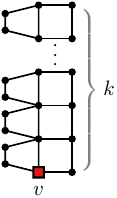}
		\subcaption{}
		\label{fig:nph:tower-gadget-a}
	\end{subfigure}
	\hfill
	\begin{subfigure}[t]{.12\textwidth}
		\centering
		\includegraphics[page=4]{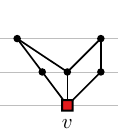}
		\subcaption{}
		\label{fig:nph:tower-gadget-b}
	\end{subfigure}
	\hfill
	\begin{subfigure}[t]{.12\textwidth}
		\centering
		\includegraphics[page=5]{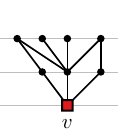}
		\subcaption{}
		\label{fig:nph:tower-gadget-c}
	\end{subfigure}
	\hfill
	\begin{subfigure}[t]{.12\textwidth}
		\centering
		\includegraphics[page=6]{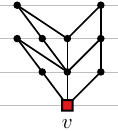}
		\subcaption{}
		\label{fig:nph:tower-gadget-d}
	\end{subfigure}
	\hfill
	\begin{subfigure}[t]{.12\textwidth}
		\centering
		\includegraphics[page=7]{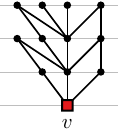}
		\subcaption{}
		\label{fig:nph:tower-gadget-e}
	\end{subfigure}
	\hfill
	\begin{subfigure}[t]{.22\textwidth}
		\centering
		\includegraphics[page=8]{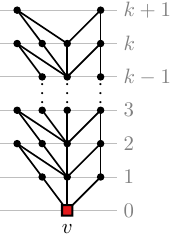}
		\subcaption{}
		\label{fig:nph:tower-gadget-f}
	\end{subfigure}
	\caption{Tower gadget of size~$k$.
		(a) shows the graph, (b)--(f) shows the unique (arched) leveled-planar drawing
		if $v$ is on its unique level and no edge is represented as an arch.}
	\label{fig:nph:tower-gadget}
\end{figure}

\begin{lemma}
    \label{lem:tower-gadget}
    If $v$ is the only vertex of the tower gadget $T_k$ of size~$k$ on its level and no edge is represented as an arch, then $T_k$ requires at least $k + 2$ levels in any leveled-planar drawing.
\end{lemma}

\begin{proof}
	We assume that $v$ is the only vertex of~$T_k$ on level~0.
	Clearly, $v$'s neighbors in the tower gadget are on the same level (i.e., level~$1$ or $-1$);
	otherwise, there would be a second vertex on level~0 as the gadget is biconnected.
	We assume, without loss of generality, that these neighbors are on level~1;
	since there are no arches allowed,
	going up or down makes no difference
    for the realizability.
	To close the first main 4-cycle and the first additional 4-cycle,
	the two vertices opposite to~$v$ are on level~2
	(see \cref{fig:nph:tower-gadget-b}).
	Furthermore, note that the middle vertex on level~1 has degree~5
	and all its undrawn neighbors must be placed one level above on level~2
	between the two vertices already assigned to level~2 to avoid crossings
	(see \cref{fig:nph:tower-gadget-c}).
	Now to close the second main cycle and the second additional cycle,
	the missing vertex of each of these cycles can only lie on level~3
	(see \cref{fig:nph:tower-gadget-d}).
	Now the argument repeats because the degree-5 vertex of level~2
	has its two undrawn neighbors between the already drawn vertices
	one level above on level~3
	(see \cref{fig:nph:tower-gadget-e}).
	The first main and additional 4-cycles occupy three levels
	while, for each consequent pair of main and additional 4-cycle,
	we require one more level.
	This results in an (arched) leveled-planar drawing on $k + 2$ levels
	(see \cref{fig:nph:tower-gadget-f}).
\end{proof}

\restateinbody{\nphardness*}
\label{thm:nphardness*}

\begin{proof}
    Clearly, the problem is in NP.
    A linear ordering of the vertices serves as a polynomial-size certificate,
    for which it can be determined in polynomial time whether
    it is a queue layout of the input graph.
    
    To show NP-hardness, we reduce from \textsc{3-Partition}.
    Given an instance $S$ of \textsc{3-Partition},
    we construct the following graph~$G$;
    see \cref{fig:nph:outerplanar,fig:nph:arched-leveled-planar} for an overview.
    We start with an anchor gadget
    and we attach to each of~$a$ and~$b$
    a copy of a (sub)graph~$G'$
    depending on the \textsc{3-Partition} instance.
    We only describe it with respect to~$a$.
    Essentially, we attach several (mostly odd-length) cycles to~$a$.

    For every $s \in S$, we have a \emph{number gadget};
    see the blue vertices and edges with orange background on the left in \cref{fig:nph:outerplanar}.
    The number gadget of~$s$ will represent the number $s$ by requiring $s$ arches on adjacent levels.
    It consists of $s$ edge-disjoint odd-lengths cycles
    that share only the vertex~$a$.
    The cycles are ordered and have lengths
    $2m(B+2) - 3$, $2m(B+2) - 5$, $2m(B+2) - 7$, \dots
    Consider a pair of consecutive cycles in the same number gadget.
    Connect two vertices neighboring $a$ within these two cycles
    by a path of length~2; see \cref{fig:nph:outerplanar} (brown edges and vertices).
    Do this in such a way that all length-2 paths are edge-disjoint.
    
    We will separate number gadgets that correspond to numbers summing to~$B$ by \emph{pocket gadgets};
    see the black and green vertices and edges with yellow background on the right in \cref{fig:nph:outerplanar}.
    For each $i \in [m]$, there is a pocket gadget,
    which contains two edge-disjoint odd-length cycles~$C_i^1$ and~$C_i^2$ (colored black)
    of lengths $2i(B+2) - 1$ and $2i(B+2) + 1$, respectively.
    Moreover the pocket gadget consists of three (green) tower gadgets of size $i(B+2)$;
    one shares an edge with $C_i^1$ and an edge with $C_i^2$,
    one shares an edge with only $C_i^1$, and
    one shares an edge with only $C_i^2$;
    see \cref{fig:nph:outerplanar} on the bottom of the pocket gadgets on the right side.
    
    This finishes the description of~$G$ and~$G'$.
    We next prove that $G$ has pathwidth at most~5.
    Afterwards, we show that $G$
    admits an arched leveled-planar drawing
    if and only if $S$ is a yes-instance.
    
    \subparagraph*{Constant Pathwidth.}
    
    We first show that $G$ has pathwidth at most~5.
    To this end, we describe path decomposition
    with at most six vertices per bag.
    Refer to \cref{fig:nph:outerplanar}.
    
    In the first sequence of bags, we cover the copy of~$G'$ attached to~$a$,
    in the second sequence, we cover the anchor gadget,
    in the third sequence we cover the copy of~$G'$ attached to~$b$.
    In the first sequence, we keep~$a$ in all bags.
    In the third sequence, we keep~$b$ in all bags.
    
    Observe that the outer boundary of a number gadget
    of a number~$s$ attached to~$v \in \{a, b\}$
    is a cycle that starts in~$v$,
    then traverses the (blue) odd-length cycles one after the other
    and between each two consecutive odd-length cycles,
    there are the (brown) path of length~2.
    Observe that having~$v$ and each pair
    of neighboring vertices in a bag~-- in the order
    along the outer boundary~--
    covers all edges of the number gadget.
    Hence, each number gadget has pathwidth~2.
    
    Now consider the ($i$-th) pocket gadget
    and let it be attached to~$v \in \{a, b\}$.
    We have~$v$ in all bags.
    We start by traversing the (green) tower gadget
    that shares a (black) vertex~$x$ with $C_i^2$.
    While traversing, we keep also~$x$ in the bag.
    We first traverse pairs of vertices 
    along the additional 4-cycle of the first main 4-cyle.
    For that, we have at most 4 vertices in the bag.
    After that we have all four vertices of
    the first main 4-cycle in a bag.
    We keep them and continue with pairs of
    consecutive vertices of the additional
    4-cycles of the second 4-cycle.
    This gives us six vertices in the bag.
    After that we have all four vertices of
    the second main 4-cycle in the bag (still six vertices).
    Now we can remove the two vertices
    shared between the first and second
    main 4-cycle.
    We repeat this process with the rest of the tower gadget.
    With at most six vertices per bag, we have covered
    all edges of that tower gadget.
    We then traverse $C_i^2$ with pairs of consecutive vertices
    along the boundary (at most three vertices per bag).
    Now we have reached the (black) vertex~$y$
    shared between $C_i^2$ and the next tower gadget,
    call it~$T$.
    When traversing~$T$, we keep
    the (black) vertex~$z$ shared between~$T$ and~$C_i^1$.
    Now traverse~$T$ in the same way as described previously
    for the tower gadget (with keeping~$z$ instead of~$x$).
    After that, we traverse~$C_i^1$, and then the remaining
    tower gadget.
    For the whole pocket gadget we have at most six vertices in a bag,
    i.e., pathwidth at most~5.
    
    So to cover the two copies of~$G'$,
    we use the previously
    described sequence for each number gadget,
    then for each pocket gadget.
    
    Finally, we describe the sequence in the middle
    representing the anchor gadget. We have the following sequence of bags:
    $(\{a, f, g\}, \{a, e, f\}, \{a, b, e\}, \{b, d, e\}, \{b, d, e\})$.
    This is at most three vertices per bag (pathwidth~2).
    
    This finishes the description
    of the path decomposition,
    which has at most six vertices per bag.

    \begin{figure}[p]
        \hspace*{-10ex}\includegraphics[page=1]{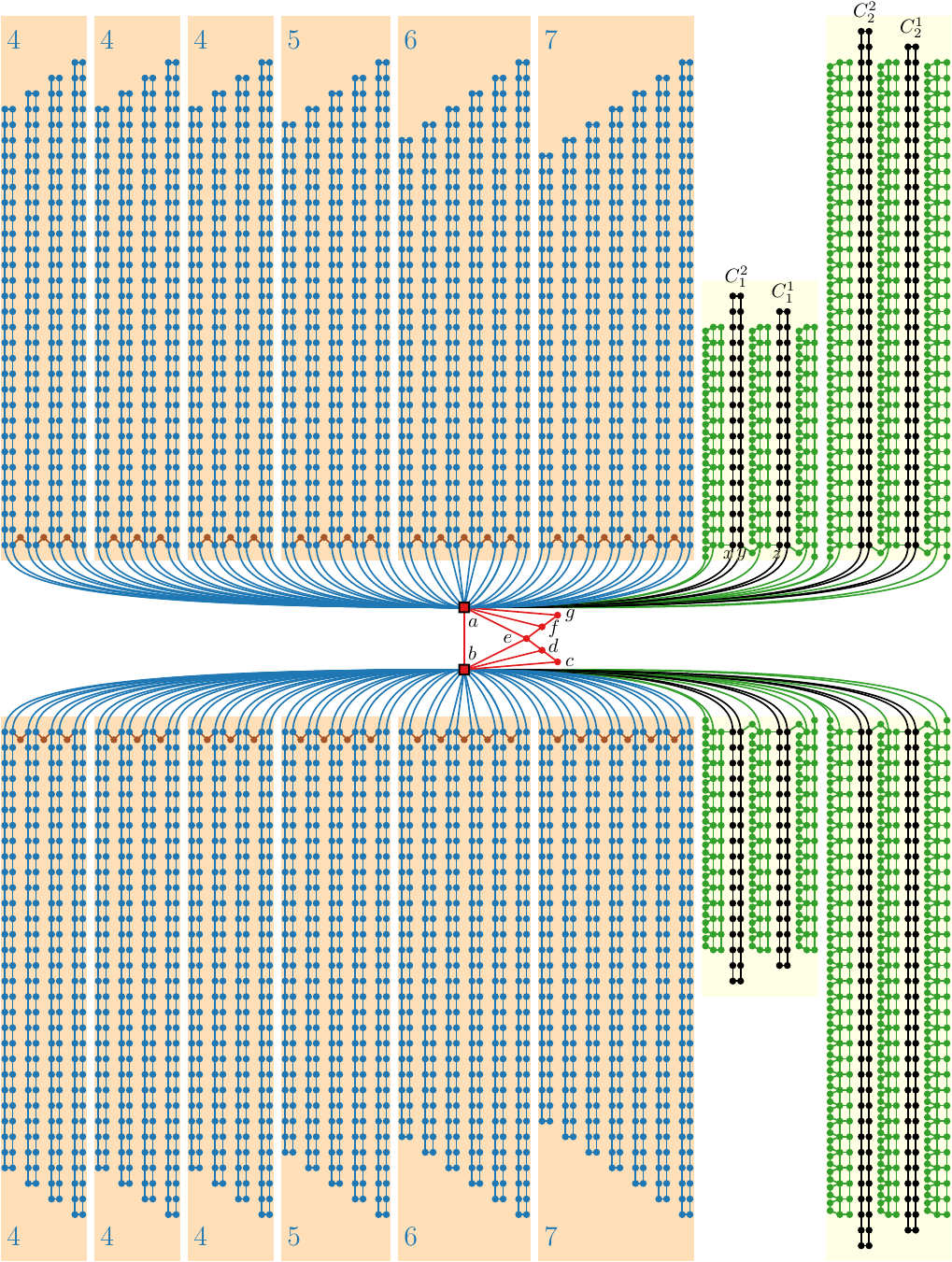}
        \caption{Outerplanar drawing for the instance $S = \{4, 4, 4, 5, 6, 7\}$.}
        \label{fig:nph:outerplanar}
    \end{figure}
    
    \begin{figure}[p]
        \hspace*{-20.5ex}\includegraphics[page=2]{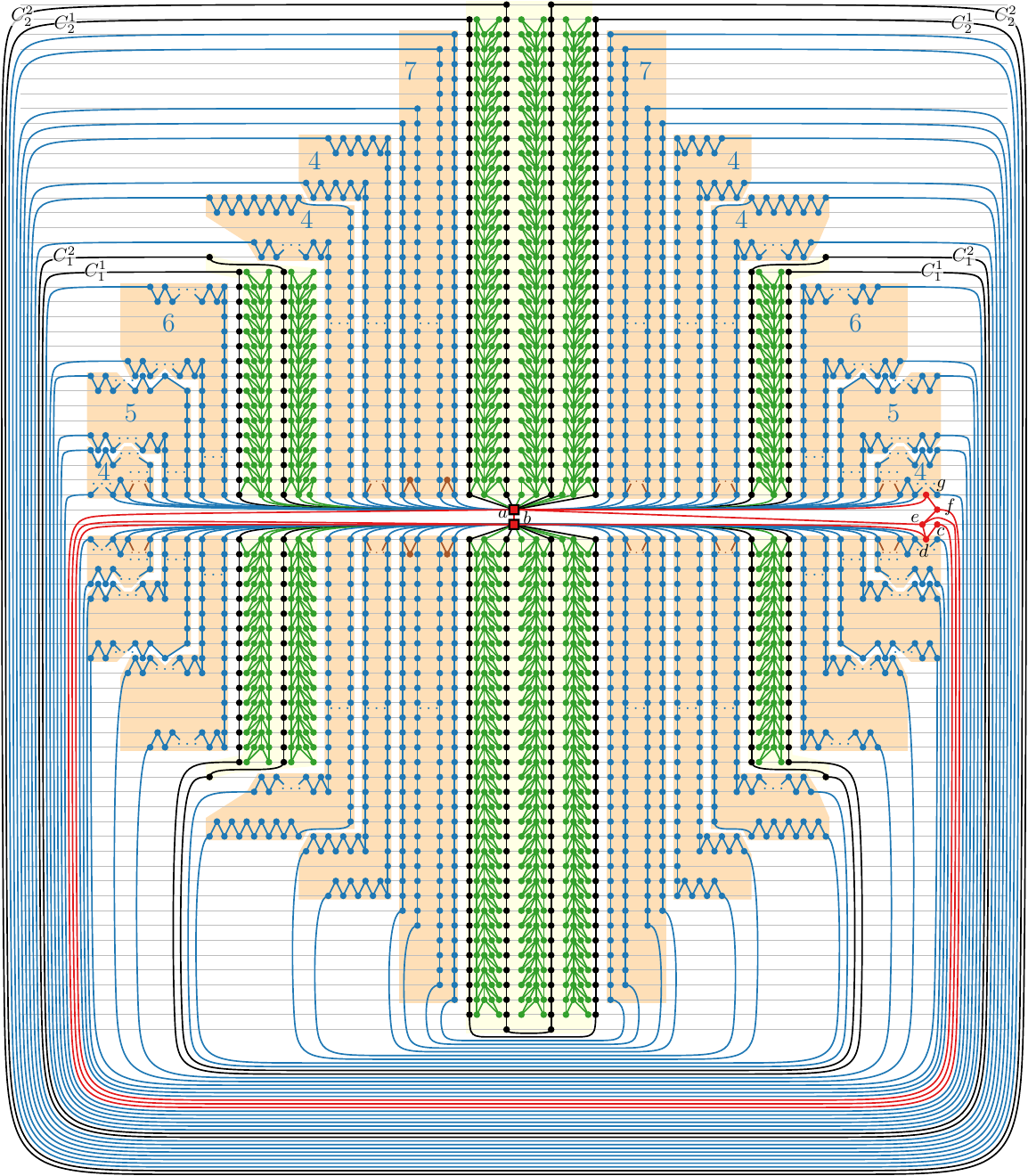}
        \caption{Arched leveled-planar drawing of the instance $S$ from \cref{fig:nph:outerplanar}.}
        \label{fig:nph:arched-leveled-planar}
    \end{figure}
    
    \subparagraph*{\textsc{3-Partition} of~$S$ $\Rightarrow$ arched leveled-planar drawing~$\Gamma$ of~$G$.}
    Suppose there is a solution $\{S_1, \dots, S_m\}$
    to the given \textsc{3-Partition} instance~$S$, i.e.,
    for each $i \in [m]$, 
    $S_i$ contains three numbers that sum to~$B$
    and $S_1, \dots, S_m$ together contain every element from~$S$.
    This allows us to construct an arched leveled-planar drawing~$\Gamma$ of~$G$.
    For an overview see \cref{fig:nph:arched-leveled-planar}.
    
    Draw the anchor gadget as in \cref{fig:nph:anchor-gadget-f}
    such that $a$ is on level~0.
    
    Next we describe the drawing of the pocket gadgets.
    For each $i \in [m]$, let $C_i^1$ and~$C_i^2$
    arch on levels $i(B+2) - 1$ and $i(B+2)$, respectively.
    Observe that this is exactly possible if each such cycle ``goes up'' at~$a$ with $i(B+2) - 1$ ($i(B+2)$, resp.) edges,
    then arches and then ``goes down'' with the remaining $i(B+2) - 1$ ($i(B+2)$, resp.) edges.
    Moreover, note that the three tower gadgets of the $i$-th pocket gadget can be drawn without arches
    going up to level~$i(B+2) - 1$ next to $C_i^1$ and~$C_i^2$~--
    the one that shares edges with both between them on one side
    and the other two on the other side somewhere next to the edges of~$C_i^1$ and~$C_i^2$.
    On the shared levels, $C_i^1$ has the outer and $C_i^2$ the inner vertices;
    see \cref{fig:nph:arched-leveled-planar}.
    
    For each $i \in [m]$, let $S_i = {s_i^1, s_i^2, s_i^3}$.
    Recall that the corresponding number gadgets have $s_i^1$, $s_i^2$, and $s_i^3$
    odd-lengths cycles, respectively.
    For brevity of notation, define $\ell = (i - 1) (B+2)$.
    We let the number gadget of $s_i^1$ arch on levels
    $\ell + 1, \ell + 2, \dots, \ell + s_i^1$,
    we let the number gadget of $s_i^2$ arch on levels
    $\ell + s_i^1 + 1, \ell + s_i^1 + 2, \dots, \ell + s_i^1 + s_i^2$, and
    we let the number gadget of $s_i^3$ arch on levels
    $\ell + s_i^1 + s_i^2 + 1, \ell + s_i^1 + s_i^2 + 2, \dots, \ell + B$.
    Observe that this is precisely between the levels where the $(i-1)$-th pocket gadget (if existent)
    and the $i$-th pocket gadget are arching~--
    their odd lengths cycles arch on levels
    $\ell - 1$, $\ell$, $\ell + B + 1$, and $\ell + B + 2$.
    
    It remains to discuss the full drawing of the odd-lengths cycles
    in the number gadgets since, unlike the pocket gadgets,
    their length does not precisely match the level where they arch.
    First note that every odd-length cycle in a number gadget
    can reach the level where it arches:
    The topmost level of arching for $s_m^3$ is $m (B+2) - 2$.
    The longest odd-length cycle of the number gadget of~$s_m^3$
    has length~$2m (B+2) - 3$, which is two vertices for
    each level in $[m (B+2) - 2]$, together with one vertex on level~0, which is~$a$.
    Clearly, the next odd-length cycle of the number gadget of~$s_m^3$
    can reach up to level~$m (B+2) - 3$ etc.
    The odd-length cycles of all other number gadgets are longer than needed.
    We can always fit this extra length into the drawing by adding zigzag
    parts of edges going up and down; see \cref{fig:nph:arched-leveled-planar}.
    Note that the extra length is always an even number~$k$,
    which perfectly fits going up $k/2$ times and going down $k/2$
    times in the zigzag parts.
    We can always fit the zigzag parts such that they do not enter level~0;
    for the first odd-length cycle of the number gadget of~$s_1^1$
    we can zigzag between levels 1 and~2, which does not prevent
    the arching of the second odd-length cycle of the number gadget.
    
    The order of the neighbors of~$a$ from left to right are
    half of the vertices in the number gadgets of~$s_1^1$, $s_1^2$, $s_1^3$,
    then half of the vertices belonging to the first pocket gadget,
    then half of the vertices in the number gadgets of~$s_2^1$, $s_2^2$, $s_2^3$,
    then half of the vertices belonging to the second pocket gadget, etc.
    In the ``middle'' we have all vertices of the $m$-th pocket gadget,
    then the other half of the vertices in the number gadgets of~$s_m^1$, $s_m^2$, $s_m^3$,
    then the other half of the vertices in the $(m-1)$-th pocket gadget, etc.,
    until we have reached the other half of the vertices
    in the number gadgets of~$s_1^1$, $s_1^2$, $s_1^3$.
    Also observe that the (brown) connecting paths of length two
    between each two consecutive odd-length cycles in the number gadgets
    can be drawn by placing their middle vertex on level~2;
    they can be drawn either on the left or the right half of the drawing.
    
    The drawing of the second copy of~$G'$ attached to~$b$
    is the same but mirrored: $b$ lies on level~$-1$
    and the rest of construction on levels $-2$ to $- m (B+2) - 2$.
    Also note that the arching conditions can be in both parts observed:
    the endpoints of the edges arching on levels $1, 2, \dots$, and $-2, -3, \dots$
    can be placed as the leftmost and the rightmost vertices of their levels,
    respectively; see \cref{fig:nph:arched-leveled-planar}.
    
    \subparagraph*{Arched leveled-planar drawing~$\Gamma$ of~$G$ $\Rightarrow$ \textsc{3-Partition} of~$S$.}
    Suppose there is a arched leveled-planar drawing~$\Gamma$ of~$G$.
    We next describe how to obtain a solution of~$S$ by~$\Gamma$.
    
    By \cref{lem:anchor-gadget}, the anchor gadget of~$G$ is drawn in one of six ways
    and there is a $v \in \{a, b\}$ such that all cycles in the number and pocket gagdets
    that are attached to~$v$ are drawn above~$v$.
    W.l.o.g., assume $v = a$ and that the level of~$a$ is~0.
    From now on, we only consider the drawing of~$G'$ attached to~$a$.
    Due to both \cref{lem:anchor-gadget} and \cref{obs:k-odd-length-cycles-k-arches},
    every odd-length cycle of the number gadgets and of the pocket gadgets
    has an arch on a distinct level greater than~0.
    All number gadgets together have $\sum_{s \in S} s = m B$ odd lengths cycles
    and the pocket gadgets together have $2m$ odd length cycles.
    Hence, there are $m (B + 2)$ levels needed for arching.
    Since the longest such cycle has length $2m (B+2) + 1$,
    it can reach at most level $m (B + 2)$ and, therefore,
    on every level from 1 to $m (B + 2)$ there is an edge arching.
    
    Consider the $i$-th pocket gadget.
    As all levels up to level~$m (B + 2)$
    are required for the odd-length cycles to arch,
    no tower gadget within a pocket gadget can arch.
    Hence, by \cref{lem:tower-gadget}, the three tower gadgets
    reach up to level $i(B+2) - 1$.
    Consider the tower gadget~$T$ that shares an edge with both~$C_i^1$ and~$C_i^2$.
    Clearly, $C_i^1$ and $C_i^2$ ``rise'' on the left and the right side of~$T$
    and then arch at vertices $u_1$ and $u_2$, respectively.
    If $u_1$ and $u_2$ are the leftmost vertices on their levels,
    then one of $u_1$ and $u_2$ must be at least on level~$i(B+2)$,
    which is one level above the tower gadget~--
    this can only be $u_2$ because only $C_i^2$ is sufficiently long.
    Then, $C_i^1$ has its arch on level~$i(B+2) - 1$
    since no other odd-length cycle of a pocket or a number gadget
    can arch in between and there needs to be an arch on every level.
    Now consider the case that $u_1$ and $u_2$ are the right endpoints of arches.
    Suppose that one of $u_1$ and $u_2$ lies to the left of
    the tower-gadget vertices on level~$i(B+2) - 1$ and the other one is further below.
    This is not yet a contradiction for the corresponding cycle to arch,
    however, on the left side, the other halves of $C_i^1$ and $C_i^2$
    have the reversed order, and they are both incident to individual copies
    of tower gadgets of size~$i(B+2)$
    reaching up to level~$i(B+2) - 1$.
    Thus, the other endpoint of edge arching on level~$i(B+2) - 1$
    is not the leftmost vertex on its level excluding this case.
    Therefore, $C_i^1$ and $C_i^2$ arch on levels $i(B+2) - 1$ and $i(B+2)$, respectively.
    Due to their length, they can not arch any higher.
    
    Knowing exactly where the odd-length cycles of the pocket gadgets arch,
    we also know that we have exactly $B$ free levels to arch
    for the number gadgets between each two neighboring pocket gadget.
    The arches of the same number gadget cannot lie above and below
    the arches of the $i$-th pocket gadget for some $i \in [m]$
    because the number gadgets are connected (by the (brown)
    length-2 paths and not just via~$a$)
    and the cycles $C_i^1$ and~$C_i^2$ act as separators in a planar drawing.
    
    For $i \in [m]$ call the levels between the $(i-1)$-th and
    the $i$-th pocket gadget the \emph{$i$-th pocket}.
    By the previous observation, a number gadget of a number~$s$ can be assigned
    to precisely one pocket and requires $s$ levels to arch.
    Since there must be an arch on every level,
    there are $B$ arches in the $i$-th pocket that belong
    to three number gadgets of numbers~$s_i^1$, $s_i^2$, and~$s_i^3$
    where $B = s_i^1 + s_i^2 + s_i^3$.
    This precisely defines a 3-partition
    $\{\{s_1^1, s_1^2, s_1^3\}, \{s_2^1, s_2^2, s_2^3\}, \dots, \{s_m^1, s_m^2, s_m^3\}\}$ of~$S$.
\end{proof}

\restateinbody{\paranphardness*}
\label{cor:parahardness-treewidth*}

We remark that \Cref{cor:parahardness-treewidth} is tight in the sense that graphs of treewidth~$1$, i.e., trees, have stack and queue number~$1$~\cite{2004layouts}.

\section{Linear-Time Recognition Algorithm for Maximal Outerplanar Graphs}
\label{sec:linear-time}

In this section, we show that 1-queue layouts of maximal outerplanar graphs can be recognized in linear time.
We first prove that every maximal outerplanar graph that admits a 1-queue layout is an outerpath.
We then analyze 1-queue layouts of maximal outerpaths along the path of the weak dual and show that, at each step, only four local states are possible and the next vertex has a unique position, which yields the recognition algorithm.
Finally, we reformulate the criterion in terms of a binary encoding of the outerpath, which we use for the structural results in \cref{sub:maximal-outerpaths}.

\subsection{Reduction to Maximal Outerpaths}

We first reduce the problem to maximal outerpaths by showing that all other maximal outerplanar graphs contain an obstruction.

\restateinbody{\maxouterplanar*}
\label{thm:maxouterplanar*}
\begin{proof}
    A maximal outerplanar graph is not an outerpath if and only if it contains the triangular grid $\mathrm{Tr}_2$ as a subgraph; see \cref{fig:no_outer}. We will now prove that there is no 1-queue layout of $\mathrm{Tr}_2$.

\begin{figure}[tb]
    \centering
    \begin{subfigure}[t]{.32\textwidth}
        \includegraphics[page=1]{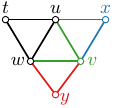}
        \subcaption{}
        \label{fig:no_outer}
    \end{subfigure}
    \hfill
    \begin{subfigure}[t]{.65\textwidth}
        \includegraphics[page=2]{figures/maximal-no.pdf}
        \subcaption{}
        \label{fig:1_queue}
    \end{subfigure}

    \medskip
    
    \begin{subfigure}[t]{\textwidth}
        \includegraphics[page=3]{figures/maximal-no.pdf}
        \subcaption{}
        \label{fig:1_queue-Q1}
    \end{subfigure}

    \medskip
    
    \begin{subfigure}[t]{\textwidth}
        \includegraphics[page=4]{figures/maximal-no.pdf}
        \subcaption{}
        \label{fig:1_queue-Q2}
    \end{subfigure}

    \medskip
    
    \begin{subfigure}[t]{\textwidth}
        \includegraphics[page=5]{figures/maximal-no.pdf}
        \subcaption{}
        \label{fig:1_queue-Q3}
    \end{subfigure}
    \caption{(a) The triangular grid $\mathrm{Tr}_2$ (b) $1$-queue layout of $\mathrm{Tr}_2\setminus\{y\}$ (c)--(e) vertex $y$ cannot appear in the $1$-queue layout of \ref{q:1}--\ref{q:3}.}
    \label{fig:maximal_no}
\end{figure}

    We first enumerate the $1$-queue layouts of $\mathrm{Tr}_2\setminus\{y\}$. Up to symmetry, assume that $t \prec v$. By \cref{lem:four-local-states,lem:unique-continuation}, the subgraph $\mathrm{Tr}_2\setminus\{x,y\}$ admits exactly four $1$-queue layouts (see the first row of \cref{fig:maximal-outerpath-states}). Reintroducing $x$ extends three of them to a layout of $\mathrm{Tr}_2\setminus\{y\}$, while the fourth admits no valid extension; in each extension, the position of $x$ is uniquely determined. 
    Thus, $\mathrm{Tr}_2\setminus\{y\}$ has exactly three $1$-queue layouts:
    \begin{enumerate*}[label={\textcolor{lipicsGray}{\bf\textsf{(Q.\arabic*)}}},ref={\textsf{Q.\arabic*}}]
    \item\label{q:1} $w \prec t \prec v \prec u \prec x$
    \item\label{q:2} $t \prec w \prec u \prec v \prec x$
    \item\label{q:3} $t \prec u \prec w \prec x \prec v$
\end{enumerate*} ; see \cref{fig:1_queue}.

    Now suppose, for a contradiction, that $\mathrm{Tr}_2$ admits a $1$-queue layout, and let $\prec$ denote its vertex ordering.  
    We will prove that for each of the orders \ref{q:1}--\ref{q:3}, an edge of vertex $y$ is involved in a nesting, which leads to a contradiction.
    \noindent{\textcolor{lipicsGray}{\bf\textsf{(Q.1)}}}: If $y\prec w$, then the edge $(y,v)$ nests the edge $(w,t)$; if $w\prec y\prec u$, then the edge $(w,u)$ nests the edge $(y,v)$; and finally, if $u\prec y$, then the edge $(w,y)$ nests the edge $(v,u)$; see \cref{fig:1_queue-Q1}.
    Similarly, {\textcolor{lipicsGray}{\bf\textsf{(Q.2)}}}: If 
    $y \prec w$, then edge $(y,v)$ nests edge $(w,u)$; if $w \prec y \prec u$, then edge $(t,u)$ nests edge $(w,y)$; if $u \prec y \prec x$, then edge $(u,x)$ nests edge $(y,v)$; if $x \prec y$ then edge $(w,y)$ nests edge $(v,x)$; see \cref{fig:1_queue-Q2}.
    Finally, {\textcolor{lipicsGray}{\bf\textsf{(Q.3)}}}: If $y \prec u$, then edge $(y,v)$ nests edge $(u,w)$; if $u \prec y \prec v$, then edge $(u,v)$ nests edge $(w,y)$; if $v \prec y$, then edge $(w,y)$ nests edge $(x,v)$; see \cref{fig:1_queue-Q3}.
\end{proof}

\subsection{Recognition Algorithm for Maximal Outerpaths}
\label{sub:rec-max-outerpaths}

We now describe the recognition algorithm for maximal outerpaths.
The key observation is that in a $1$-queue layout, the local order of each triangle and its predecessor can attain only four combinatorial states, and each new vertex has a uniquely determined position.

Let $G$ be a maximal outerpath with $n$ vertices and let $F_0,\dots,F_{n-3}$ be its inner triangles,
ordered along the path of the weak dual.
For $i<n-3$, we call the edge shared by $F_i$ and $F_{i+1}$ the \emph{frontier edge} at step~$i$.

Consider a step $i\ge 1$.
We label the vertices of $F_{i-1}$ and $F_i$ as follows:
$t_i$ is the vertex of $F_{i-1}$ that does not belong to $F_i$,
$v_i$ is the vertex of $F_i$ that does not belong to $F_{i-1}$,
and $u_i,w_i$ are the endpoints of the edge shared by $F_{i-1}$ and $F_i$,
named so that the frontier edge at step~$i$ is $(u_i,v_i)$.
(For $i=n-3$, there is no frontier edge and the labels $u_{n-3},w_{n-3}$ can be chosen arbitrarily.)
We call $(t_i,u_i,w_i,v_i)$ the \emph{frontier quadruple} at step~$i$;
it induces the graph $K_4-(t_i,v_i)$.

We first analyze which linear orders of a frontier quadruple are admissible. 
Since the analysis is local, we drop the index~$i$.

\begin{lemma}
  \label{lem:four-local-states}
  Let $t,u,w,v$ be four vertices inducing $K_4-(t,v)$.  If we assume
  that $t\prec v$, then a linear order of $t,u,w,v$ contains no two
  nested edges if and only if it is one of the following four
  \emph{states}:
\begin{description}
\item[$-1$:] $w\prec t\prec v\prec u$,
\item[$\phantom{-}0$:] $t\prec w\prec u\prec v$,
\item[$\phantom{-}1$:] $t\prec u\prec w\prec v$,
\item[$\phantom{-}2$:] $u\prec t\prec v\prec w$.
\end{description}
\end{lemma}

\begin{proof}
Among four vertices $a\prec b\prec c\prec d$, the only pair of edges that can nest is $(a,d)$ and $(b,c)$.
Since every pair in $\{t,u,w,v\}$ except $\{t,v\}$ is an edge,
the order is nesting-free if and only if $\{t,v\}$ occupies one of these two position pairs,
that is, if and only if $t$ and $v$ are the two outermost or the two middle vertices.
Under the assumption $t\prec v$, this yields exactly the four listed orders.
\end{proof}

The assumption $t\prec v$ is without loss of generality for the first frontier quadruple,
since reversing a 1-queue layout yields a 1-queue layout.
We will see in Lemma~\ref{lem:unique-continuation} that every subsequent frontier quadruple then satisfies $t_i\prec v_i$ automatically,
so it suffices to work with the four states $-1,0,1,2$ throughout.

Now let $x$ be the vertex of $F_{i+1}$ that does not belong to $F_i$,
so $F_{i+1}=(u_i,v_i,x)$ is attached to the frontier edge $(u_i,v_i)$.
The frontier edge at step $i+1$ is either $(u_i,x)$ or $(v_i,x)$;
we call the step a \emph{same move} in the first case and an \emph{opposite move} in the second;
see Figure~\ref{fig:maximal-outerpath-states}.
By the labeling convention, the next frontier quadruple is
\[
(t_{i+1},u_{i+1},w_{i+1},v_{i+1})=
\begin{cases}
(w_i,\,u_i,\,v_i,\,x) & \text{for a same move,}\\
(w_i,\,v_i,\,u_i,\,x) & \text{for an opposite move.}
\end{cases}
\]

Throughout the construction, we maintain the following invariant:
\begin{description}
    \item[(I)] The two rightmost vertices of the frontier quadruple are the two rightmost vertices of the entire partial layout.
\end{description}
For the first quadruple, only four vertices have been placed, so (I) holds trivially.

\begin{lemma}
\label{lem:unique-continuation}
If the current frontier configuration is in one of the states
$-1,0,1,2$,
then the position of the next vertex $x$ is uniquely determined.
The new frontier quadruple satisfies $t_{i+1}\prec v_{i+1}$ and invariant~(I),
and its state is given by the following table.
\[
\begin{array}{c|c|c}
\text{current state} & \text{same move} & \text{opposite move}\\
\hline
-1\phantom{-} & 0 & 1 \\
0  & 1      & 0\\
1  & 2      & -1\phantom{-}\\
2  & \text{\xmark} & \text{\xmark}
\end{array}
\]
\end{lemma}

In other words, a same move increases the state by one and an opposite move negates it;
state~$2$ is a dead end.

\begin{figure}
    \centering
    \includegraphics[]{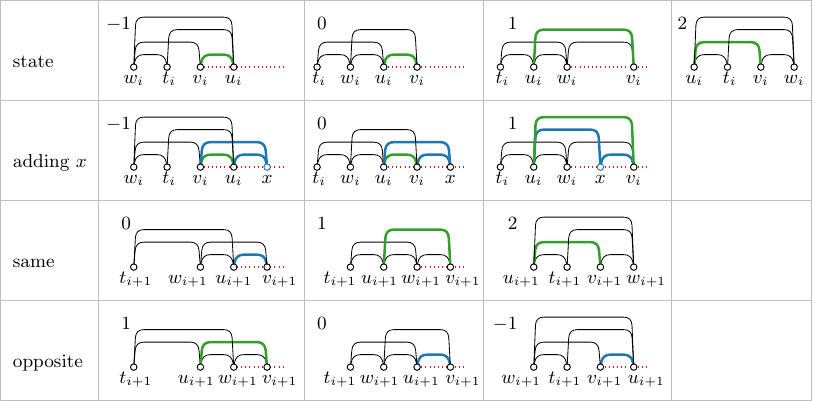}
    \caption{Illustration for Lemma~\ref{lem:unique-continuation}. From top to bottom: The state, placement of $x$, relabeling after a same move, relabeling after an opposite move. Dotted regions contain no other vertex.}
    \label{fig:maximal-outerpath-states}
\end{figure}

\begin{proof}
We inspect the four states one by one and analyze the possible gaps for $x$;
see Figure~\ref{fig:maximal-outerpath-states}.
Note first that in every transition we have $t_{i+1}=w_i$ and $v_{i+1}=x$,
and in every admissible placement below $x$ lies to the right of $w_i$;
hence $t_{i+1}\prec v_{i+1}$ always holds.

\begin{description}

\item[State $-1$: $w_i\prec t_i\prec v_i\prec u_i$.]~

\medskip
  
\begin{tabularx}{\linewidth}{@{}>{\normalsize}l>{\normalsize}X@{}}
$x\prec w_i$:     & Edge $(u_i,x)$ nests $(w_i,v_i)$. \xmark \\
$w_i\prec x\prec t_i$: & Edge $(u_i,w_i)$ nests $(x,v_i)$. \xmark \\
$t_i\prec x\prec u_i$: & Edge $(t_i,u_i)$ nests $(x,v_i)$. \xmark \\
$u_i\prec x$:     & Admissible.
By (I), no vertex lies to the right of $u_i$,
so this is a single gap and $x$ becomes the rightmost vertex.
The interval between $u_i$ and $x$ is empty and the interval between $v_i$ and $x$ contains only $u_i$,
so the new edges $(u_i,x)$ and $(v_i,x)$ are not involved in any nesting. 
\end{tabularx}

\medskip

For a same move, the new labels are
$(t_{i+1},u_{i+1},w_{i+1},v_{i+1})=(w_i,u_i,v_i,x)$.
In the current order, we have $t_{i+1}\prec w_{i+1}\prec u_{i+1}\prec v_{i+1}$,
so we arrive at state $0$,
and $u_{i+1}=u_i$ and $v_{i+1}=x$ are the two rightmost vertices.

For an opposite move, the new labels are
$(t_{i+1},u_{i+1},w_{i+1},v_{i+1})=(w_i,v_i,u_i,x)$.
In the current order, we have $t_{i+1}\prec u_{i+1}\prec w_{i+1}\prec v_{i+1}$,
so we arrive at state $1$,
and $w_{i+1}=u_i$ and $v_{i+1}=x$ are the two rightmost vertices.

\medskip
\item[State $0$: $t_i\prec w_i\prec u_i\prec v_i$.]~

\medskip

\begin{tabularx}{\linewidth}{@{}>{\normalsize}l>{\normalsize}X@{}}
$x\prec w_i$: & Edge $(x,v_i)$ nests $(w_i,u_i)$. \xmark \\
$w_i\prec x\prec v_i$: & Edge $(w_i,v_i)$ nests $(u_i,x)$. \xmark \\
$v_i\prec x$: & Admissible.
By (I), no vertex lies to the right of $v_i$,
so this is a single gap and $x$ becomes the rightmost vertex.
The interval between $v_i$ and $x$ is empty and the interval between $u_i$ and $x$ contains only $v_i$,
so the new edges $(u_i,x)$ and $(v_i,x)$ are not involved in any nesting. 
\end{tabularx}

\medskip

For a same move, the new labels are
$(t_{i+1},u_{i+1},w_{i+1},v_{i+1})=(w_i,u_i,v_i,x)$.
In the current order, we have $t_{i+1}\prec u_{i+1}\prec w_{i+1}\prec v_{i+1}$,
so we arrive at state $1$,
and $w_{i+1}=v_i$ and $v_{i+1}=x$ are the two rightmost vertices.

For an opposite move, the new labels are
$(t_{i+1},u_{i+1},w_{i+1},v_{i+1})=(w_i,v_i,u_i,x)$.
In the current order, we have $t_{i+1}\prec w_{i+1}\prec u_{i+1}\prec v_{i+1}$,
so we arrive at state $0$,
and $u_{i+1}=v_i$ and $v_{i+1}=x$ are the two rightmost vertices.

\medskip
\item[State $1$: $t_i\prec u_i\prec w_i\prec v_i$.]~

\medskip

\begin{tabularx}{\linewidth}{@{}>{\normalsize}l>{\normalsize}X@{}}
$x\prec u_i$:     & Edge $(x,v_i)$ nests $(u_i,w_i)$. \xmark \\
$u_i\prec x\prec w_i$: & Edge $(t_i,w_i)$ nests $(u_i,x)$. \xmark \\
$w_i\prec x\prec v_i$: & Admissible.
By (I), $w_i$ and $v_i$ are the two rightmost vertices,
so this gap is empty.
The edge $(v_i,x)$ cannot be involved in a nesting,
since the interval between $x$ and $v_i$ is empty and no vertex lies to the right of $v_i$.
The edge $(u_i,x)$ cannot be involved in a nesting either:
any edge that nests $(u_i,x)$ also nests $(u_i,w_i)$,
and any edge that $(u_i,x)$ nests is also nested by $(u_i,v_i)$.  \\
$v_i\prec x$:     & Edge $(u_i,x)$ nests $(w_i,v_i)$. \xmark
\end{tabularx}

\medskip

For a same move, the new labels are
$(t_{i+1},u_{i+1},w_{i+1},v_{i+1})=(w_i,u_i,v_i,x)$.
In the current order, we have $u_{i+1}\prec t_{i+1}\prec v_{i+1}\prec w_{i+1}$,
so we arrive at state $2$,
and $w_{i+1}=v_i$ and $v_{i+1}=x$ are the two rightmost vertices.

For an opposite move, the new labels are
$(t_{i+1},u_{i+1},w_{i+1},v_{i+1})=(w_i,v_i,u_i,x)$.
In the current order, we have $w_{i+1}\prec t_{i+1}\prec v_{i+1}\prec u_{i+1}$,
so we arrive at state $-1$,
and $u_{i+1}=v_i$ and $v_{i+1}=x$ are the two rightmost vertices.

\medskip
\item[State $2$: $u_i\prec t_i\prec v_i\prec w_i$.]~

\medskip

\begin{tabularx}{\linewidth}{@{}>{\normalsize}l>{\normalsize}X@{}}
$x\prec u_i$:     & Edge $(x,v_i)$ nests $(u_i,t_i)$. \xmark \\
$u_i\prec x\prec w_i$: & Edge $(u_i,w_i)$ nests $(x,v_i)$. \xmark \\
$w_i\prec x$:     & Edge $(u_i,x)$ nests $(t_i,w_i)$. \xmark
\end{tabularx}

\medskip

Hence, no admissible placement exists;
state~$2$ has no continuation.\qedhere
\end{description}
\end{proof}

\restateinbody{\lineartime*}
\label{thm:lineartime*}
\begin{proof}
By \cref{thm:maxouterplanar}, $G$ must be a maximal outerpath; otherwise, the answer is 'no'.
The dual path $F_0,\dots,F_{n-3}$ and the frontier edges can be computed in linear time.
If $n=3$, then $G$ is a single triangle and every vertex order is a 1-queue layout,
so assume $n\ge 4$.

Suppose that $G$ admits a 1-queue layout.
After possibly reversing it,
the first frontier quadruple satisfies $t_1\prec v_1$ and is thus, by Lemma~\ref{lem:four-local-states}, in one of the states $-1,0,1,2$.
											   
If $n\ge 5$, then state $2$ is excluded by Lemma~\ref{lem:unique-continuation};
if $n=4$, then $G$ is exactly $K_4-(t_1,v_1)$ and all four states are 1-queue layouts of $G$.
In both cases, $G$ admits a 1-queue layout whose first quadruple is in state $-1$, $0$, or $1$.

The algorithm therefore runs the following procedure once for each initial state $\sigma\in\{-1,0,1\}$:
place the first four vertices in the order prescribed by $\sigma$,
then process $F_2,\dots,F_{n-3}$,
at each step inserting the new vertex into the unique gap given by Lemma~\ref{lem:unique-continuation} and updating the state via the transition table;
abort if the current state is $2$ and a triangle remains.
If a run completes, its output is a 1-queue layout,
since by Lemma~\ref{lem:unique-continuation} no step creates a nested pair of edges.
Conversely, if $G$ admits a 1-queue layout whose first quadruple is in state $\sigma$,
then every extension step is forced,
so the run for $\sigma$ never aborts.
Hence, $G$ admits a 1-queue layout if and only if at least one of the three runs completes.
\end{proof}

\subsection{Characterization for Maximal Outerpaths}\label{sec:characterization}

In this subsection, we give a different perspective on the structure of maximal outerpaths that admit a 1-queue layout, which leads to an alternative linear-time recognition algorithm and helps us prove several structural results later.

Let $F_0,F_1,\dots,F_{n-3}$ be the $n-2$ triangles of the maximal outerpath,
ordered along the path of the weak dual, and, for $i \in [n-3]$,
let $e_i=F_{i-1}\cap F_i$ be the diagonal
shared by $F_{i-1}$ and $F_i$.  Draw the outerpath in a
horizontal strip with an upper and a lower boundary chain; see \cref{fig:characterization-example}.

For each $i\in[n-4]$, let $a_i$ be the \emph{apex} of $F_i$, that is, the vertex
of $F_i$ not incident to $e_i$, and set $w_i=\mathrm{U}$ if $a_i$ lies on the upper chain, and $w_i=\mathrm{D}$ otherwise. 
The word $w=w_1w_2\cdots w_{n-4}\in\{\mathrm{U},\mathrm{D}\}^{n-4}$ is the \emph{$\mathrm{U}/\mathrm{D}$-encoding} of the outerpath.
The first face $F_0$ carries no letter: $F_0$ has no preceding diagonal,
and the apex of $F_{n-3}$ is an end vertex of degree~$2$ that can be drawn on either chain,
so a letter for $F_{n-3}$ would be a property of the drawing rather than of the graph.

\begin{figure}[ht]
    \centering
    \begin{subfigure}{.65\textwidth}
    \centering
    \includegraphics[page=6]{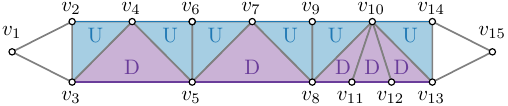}
    \caption{}
    \label{fig:characterization-example}
    \end{subfigure}
    \begin{subfigure}{.65\textwidth}
    \centering
    \includegraphics[page=7]{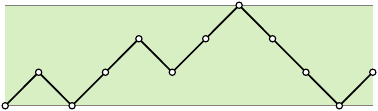}
    \caption{}
    \label{fig:characterization-path}
    \end{subfigure}
    \begin{subfigure}{.65\textwidth}
    \centering
    \includegraphics[page=8]{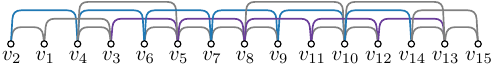}
    \caption{}
    \label{fig:characterization-drawing}
    \end{subfigure}
    \caption{(a) The maximal outerplanar graph $G$ corresponding to the $\mathrm{U}/\mathrm{D}$-encoding $\mathrm{UDUUDUUDDDU}$; (b)~its lattice path has height~3; (c)~a $1$-queue layout of~$G$.}
    \label{fig:characterization}
\end{figure}

Since $w_i$ records the side on which $F_i$ contributes its new outer
edge when crossing $e_i$, two consecutive faces lie on the same side
of the strip exactly when $w_i=w_{i+1}$.  This corresponds exactly to
the same move
in \cref{sub:rec-max-outerpaths} and yields a recognition
algorithm formulated directly in terms of the
$\mathrm{U}/\mathrm{D}$-encoding.  Given $w=w_1w_2\cdots w_{n-4}\in\{\mathrm{U},\mathrm{D}\}^{n-4}$, define the prefix
difference $d_0=0$ and, for $i \in [n-4]$,
\[
  d_i = \#\mathrm{U}(w_1\cdots w_i)-\#\mathrm{D}(w_1\cdots w_i),
\]
and read the sequence $(0,d_0),(1,d_1),\dots,(n-4,d_{n-4})$ as a \emph{lattice path}
that goes up one unit on~$\mathrm{U}$ and goes down one unit on~$\mathrm{D}$; see \cref{fig:characterization-path}.
The \emph{height} of the path is the difference between the largest and the smallest value among $d_0,\dots,d_{n-4}$.

\begin{proposition}\label{prop:span}
A maximal outerpath admits a 1-queue layout if and only if the height of its lattice path is at most $3$.
\end{proposition}
\begin{proof}
    By \cref{lem:four-local-states,lem:unique-continuation},
    the outerpath admits a 1-queue layout if and only if there is an initial state $\sigma\in\{-1,0,1\}$ such that the state sequence given by $s_1=\sigma$ and
    \[ s_{i+1}=
        \begin{cases}
            s_i+1 & \text{for a same move, that is, } w_i=w_{i+1},\\
            -s_i & \text{for an opposite move, that is, } w_i\ne w_{i+1},
        \end{cases}
    \]
    satisfies $s_i\in\{-1,0,1\}$ for every $i\in [n-4]$.

    Geometrically, choosing the initial state amounts to choosing the
    vertical position of the window of height~$3$, and the terminal
    state~$2$ is reached as soon as the path leaves the chosen window.
\end{proof}

    Since the height is computed by a single left-to-right scan of the word,
    maintaining the running minimum and maximum of the prefix differences,
    \cref{prop:span} yields a linear-time recognition algorithm.

\section{Existential Results for Outerpaths of Bounded Degree}
\label{sec:outerpaths}

In this section, we establish all four statements of \Cref{prop:max-degree}:
\restateinbody{\maxDegreeProposition*}
\label{prop:max-degree*}
We begin with Statements~1 and~2, which we show in the upcoming \Cref{sub:maximal-outerpaths} (see \cref{prop:degree-6,prop:degree-4-max-outerpath}).
In \Cref{sub:degree-3-outerpaths}, we show Statement~3 (see \Cref{thm:outerpath-max-deg-3}) and the proof of Statement~4 can be found in \Cref{sub:degree-4-outerpaths} (see \Cref{prop:max-deg-4-outerpaths}).

\subsection{Maximal Outerpaths}
\label{sub:maximal-outerpaths}

In this subsection, we use the $\mathrm{U}/\mathrm{D}$-encoding of \cref{sec:characterization} to derive structural properties of maximal outerpaths that admit a 1-queue layout.
Runs of equal letters in the encoding correspond to high-degree vertices, so the height criterion of \cref{prop:span} translates into degree bounds.

\begin{proposition}\label{prop:degree-6}
Every maximal outerpath that admits a $1$-queue layout has degree at most~$6$.
\end{proposition}
\begin{proof}
We first bound the length of a run of equal letters in the $\mathrm{U}/\mathrm{D}$-encoding of \cref{sec:characterization}.
Consider a run $w_j=w_{j+1}=\dots=w_{j+\ell-1}$ of $\ell$ equal letters. 
Along this run, the lattice path takes $\ell$ steps in the same direction,
so the values $d_{j-1},d_j,\dots,d_{j+\ell-1}$ are strictly monotone and span $\ell$.
By \cref{prop:span}, the height of the lattice path is at most $3$, so $\ell\le 3$.

Now suppose that the run above is maximal.
Throughout the run,
the endpoint of the frontier edge on the chain opposite to the apexes stays fixed;
we call this vertex the \emph{anchor} of the run.
The anchor is incident exactly to the faces $F_{j-1},\dots,F_{j+\ell}$,
and a vertex incident to exactly $r$ consecutive faces has degree $r+1$.
Hence, the anchor of a maximal run of length $\ell$ has degree $\ell+3$.
Since every vertex of degree at least $4$ is the anchor of some maximal run,
every vertex has degree at most $3+3=6$.
\end{proof}

\Cref{prop:degree-6} shows that large degree is an obstruction to 1-queue layouts.
At the other end of the spectrum, degree at most~$4$ determines the maximal outerpath uniquely, and this outerpath always admits a 1-queue layout.

\begin{proposition}\label{prop:degree-4-max-outerpath}
    For every $n$, the maximal $n$-vertex outerpath of degree at most~4 admits a $1$-queue layout. 
\end{proposition}

\begin{proof}
    In a maximal outerpath, every vertex is incident to two edges of the outer face. 
    Since the maximum degree is~4, every vertex is incident to at most two inner edges.
    Consider the $\mathrm{U}/\mathrm{D}$-encoding from \cref{sec:characterization}.
    
    Three inner edges meet in a common vertex if and only if $F_i$ and
    $F_{i+1}$ keep the same apex, that is, if and only if $w_i=w_{i+1}$. Therefore,
    $\Delta\le 4$ forbids two equal consecutive letters, and the $\mathrm{U}/\mathrm{D}$-encoding is
    strictly alternating; see~\cref{fig:maximal_outerpath_deg4_enc}.

\begin{figure}[t]
    \centering
    \begin{subfigure}[t]{.42\textwidth}
        \centering
        \includegraphics[page=4]{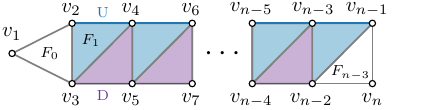}
        \subcaption{}
        \label{fig:encoding}
    \end{subfigure}
    \hfill
    \begin{subfigure}[t]{.53\textwidth}
        \centering
        \includegraphics[page=5]{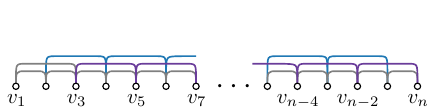}
        \subcaption{}
        \label{fig:1_queue_en}
    \end{subfigure}
    \caption{(a) The unique maximal outerpath where each vertex has degree at most~$4$ and (b) a queue layout of the graph.}
    \label{fig:maximal_outerpath_deg4_enc}
\end{figure}
    
    For a strictly alternating word, the prefix difference $d_i$ steps $+1$ and $-1$
    in turn from $d_0=0$, so all $d_i$ lie in $\{0,1\}$ or all in $\{-1,0\}$. 
    In either case, the height of the lattice path is at most $1\le 3$, and by \cref{prop:span} the outerpath is $1$-queue-embeddable.
\end{proof}

\subsection{Outerpaths with Maximum Degree~3}
\label{sub:degree-3-outerpaths}

In this section, we establish the third statement from \Cref{prop:max-degree}, i.e., that every outerpath of maximum degree~3
admits a 1-page queue layout \lQL.  

Let $G$ be such an outerpath with
$\Delta(G) \leq 3$. 
Since $G$ is an outerpath, it can be decomposed
into a Hamiltonian cycle $H$ with $V(H)=V(G)$ and
$E(H) \subseteq E(G)$, and a set of chords $C = E(G) \setminus E(H)$;
see \Cref{fig:outerpath-degree-three}a.
\begin{figure}
	\centering
	\includegraphics[page=1]{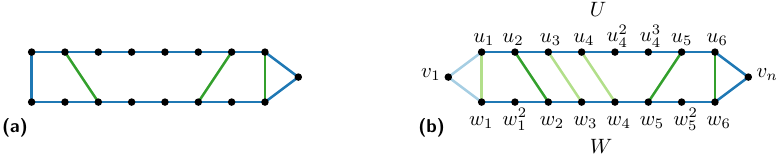}
	\caption{\textbf{\textsf{(a)}} An outerpath $G$ with maximum degree 3. We color the Hamiltonian cycle $H$ in blue and the matching $C$ in green. \textbf{\textsf{(b)}} The supergraph $G^+$ for the graph $G$. We color new edges in the Hamiltonian cycle and matching in light blue and light green, respectively.}
	\label{fig:outerpath-degree-three}
\end{figure}
Moreover, since $\Delta(G) \leq 3$, the chords $C$ form a matching in $G$, i.e., every vertex $v \in V$ is incident to at most one chord.
In the following, we let $f(G)$ denote the number of interior faces of $G$ and for every $i \in [f(G)]$, we let $F_i(G)$ denote the interior faces of $G$, ordered as they appear along the weak dual of (a planar drawing of) $G$.
Note that every pair of consecutive faces $F_i$ and $F_{i + 1}$ is separated by a chord $c_i \in C$.
First, we show a useful property about the considered outerpaths.

\begin{lemma}
	\label{lem:outerpath-deg-3-supergraph}
	Every outerpath $G$
	with $\Delta(G) \leq 3$ is a subgraph of
	an outerpath $G^+$ with $\Size{V(G^+) \setminus V(G)} \leq 2$
	and $\Delta(G^+) \leq 3$ such that $F_1$ and $F_{f(G^+)}$ are
	triangles and for every $1 < i < f(G^+)$, at least three sides of
	$F_i$ consist of a single edge, two of which are the chords~$c_{i - 1}$ and $c_i$.
\end{lemma}
\begin{proof}
	Let $G$ be an outerpath with $\Delta(G) \leq 3$.
	If $G$ is a triangle, there is nothing to show, so we assume $G$ has at least four vertices.
	Moreover, for the sake of the proof, assume that $G$ does not fulfill the properties of $G^+$ (otherwise, we can simply set $G^+ = G$ and are done).
	
	We first show how to ensure that the first and last faces are triangles.
	Consider the first face $F$ of $G$.
	If it is not a triangle, it is incident to at least four vertices, two of which are consecutive and of degree two in $G$.
	Let these vertices be $x$ and $y$.
	We introduce the vertex $a_1$, which is adjacent to precisely $x$ and $y$, i.e., in the resulting graph, $a_1$ is of degree two and $x$ and $y$ are of degree three.
	Observe that the resulting graph remains an outerpath.
	Analogous arguments can be applied to the last face $F'$ of $G$, where we introduce the vertex $a_n$ if necessary.
	In the end, we have $V(G^+) \setminus V(G) \subseteq \{a_1, a_n\}$.
	
	Let $G'$ be the resulting graph.
	If $G'$ fulfills the properties of the lemma statement, we have $G^+ = G'$ and are done.
	Otherwise, let $F_i$ be an arbitrary face that does not fulfill the statement.
	Since $F_i$ is bounded by the chords $c_{i - 1} = uw$ and $c_{i} = u'w'$, this implies that the two remaining sides of $F_i$ are bounded by two paths $P = (u = u_1, u_2, \ldots, u_p = u')$ and $Q = (w = w_1, w_2, \ldots, w_q = w')$ such that $p,q > 2$.
	Note that every vertex $u_x$ and $w_y$, $1 < x < p, 1 < y < q$, has degree two in $G'$.
	We now introduce the chords $u_jw_j$ for every $2 \leq j < \min(p,q)$; see \Cref{fig:outerpath-degree-three}b.
	Observe that they just partition the face into several smaller faces, i.e., the graph remains an outerpath.
	Moreover, since every $u_j$ and $w_j$ has degree two in $G'$, the resulting graph has maximum degree three.
	Finally, after this process, at least three sides of each resulting face consist of a single edge; recall also \Cref{fig:outerpath-degree-three}b.
	
	By iteratively applying the above steps, we eventually obtain the graph $G^+$.
	Note that this does not introduce any further vertex, i.e., it remains $V(G^+) \setminus V(G) \subseteq \{a_1, a_n\}$.
\end{proof}

In the following, we call an outerpath $G$ that results from the application of \Cref{lem:outerpath-deg-3-supergraph} a \emph{maximal} outerpath with $\Delta(G) \leq 3$.
Let $G$ be a maximal outerpath and let $V(F_i)$ denote the vertices incident to face $F_i$ for every $i \in [f(G)]$.
We let $G_i$ denote the subgraph of $G$ induced by $V(F_1) \cup \ldots \cup V(F_i)$ and use $V_i = V(G_i)$ and $E_i = E(G_i)$ as a shorthand for its vertex and edge set, respectively.
Moreover, we let $v_1$ and $v_n$ be the unique vertices of degree two in the first and last face, respectively.
Observe that we can decompose $H$ into two internally vertex-disjoint paths connecting $v_1$ with $v_n$.
We let $U$ and $W$ denote the vertex-disjoint subpaths, i.e., after dropping $v_1$ and $v_n$; see the vertex labeling in \Cref{fig:outerpath-degree-three}b.
Moreover, we let~$u_i$ and~$w_i$ denote the endpoints of the chord~$c_i$.
We use $U_i = (u_i = u_i^1, \ldots, u_i^p = u_{i + 1})$ and $W_i = (w_i = w_i^1, \ldots, w_i^q = w_{i + 1})$ to refer to the two subpaths connecting the chords.
Recall that $p = 2$ or $q = 2$ (or both) due to \Cref{lem:outerpath-deg-3-supergraph}.

Our high-level idea is as follows.
We traverse the weak dual of $G$ from left to right and (iteratively) build a one-page queue layout $\prec_i$ for the graph $G_i$.
During construction, we ensure that the chord $c_i$ is among the two ``rightmost'' edges in $\prec_i$, which allows us to extend $\prec_i$ to a one-page queue layout $\prec_{i + 1}$ for $G_{i + 1}$.

\subparagraph*{Frontiers and Compressed Layouts.}
The above-sketched idea is formalized via the concept of a frontier, which is a special pair of vertices in a linear layout.
\begin{definition}
	\label{def:outerpath-deg-3-frontier}
	Let $\prec$ be a one-page queue layout of a graph $G$.
	A \emph{frontier} of $\prec$ is a vertex pair $(u,w)$, $u,w \in V(G)$ such that
	\begin{enumerate} %
		\item $N_G(w) = \{u,z\}$ for some $z \in V(G)$, i.e., $\deg_G(w) = 2$,
		\item for every $x \in V(G)$ it holds that $x \preceq w$, i.e., $w$ is the rightmost vertex in $\prec$, and
		\item for every vertex $y \in V(G)$ with $u \prec y \prec w$, we have $x \prec u$ for every $x \in N_G(y) \setminus \{w\}$.
	\end{enumerate}
\end{definition}

\begin{figure}
	\centering
	\includegraphics[page=1]{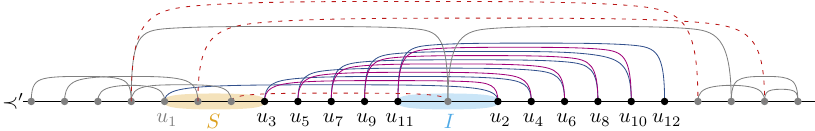}
	\caption{The compressed extension $\prec'$ of the layout $\prec$ by the path $P$ on twelve vertices. We indicate the graph $G$ in gray and the separator $S$ interior $I$ in yellow and blue, respectively. Red dotted edges are forbidden in $G$, since they are nesting with $P$.}
	\label{fig:compressed-extension}
\end{figure}

Note that frontiers are not unique, i.e., a layout could have multiple frontiers.
Later, frontiers serve as well-defined separating lines between the already-placed vertices and those that are yet-to-be-placed.
In particular, the endpoints of every chord $c_i \in C$ will constitute a frontier of some partial layout.
To ensure this, we will need an ``efficient'' way of placing the vertices of a path $P$ in~$\prec_i$.
\begin{definition}
	\label{def:compressed-extension}
	Let $\prec$ be a one-page queue layout of some graph $G$ and let $P = (u_1, \ldots, u_{p})$, $p \geq 2$, be a path with an even number of vertices such that $V(G) \cap V(P) = \{u_1\}$.
	Moreover, let $S$ be an interval directly right of $u_1$ and let $I$ be an interval directly right of $S$.
	Note that both $S$ and $I$ could be empty.
	The \emph{compressed extension} of $\prec$ by $P$ with respect to \emph{separator} $S$ and \emph{interior} $I$ is the layout $\prec'$ obtained by taking the transitive closure of
	\begin{itemize}
		\item $u_1 \prec' u_3 \prec' \ldots \prec' u_{p - 1} \prec' u_2 \prec' u_4 \prec' \ldots \prec' u_{p}$,
		\item $x \prec' y$ for every $x,y \in V(G)$ with $x \prec y$,
		\item $u_1 \prec' v_s \prec' u_3$ for every $v_s \in S$,
		\item $u_{p - 1} \prec' v_i \prec' u_2$ for every $v_i \in I$, and
		\item $u_{p} \prec' x$ for every $x \in V(G) \setminus (S \cup I)$ with $u_1 \prec x$.
	\end{itemize}
	We call $(u_1, u_3, \ldots, u_{p - 1})$ and $(u_2, u_4, \ldots, u_{p})$ the \emph{odd} and \emph{even} interval of $P$, respectively.
\end{definition}

Note that \Cref{def:compressed-extension} requires $\Size{V(P)}$ to be even; the case for $\Size{V(P)}$ being odd will be handled separately.
We make the following observation; see also \Cref{fig:compressed-extension}.

\begin{observation}
	\label{obs:compressed-extension}
	Let $P = (u_1, \ldots, u_{p})$, $p \geq 2$, be a path on an even number of vertices and let $\prec$ be a one-page queue layout of $G$ such that $V(P) \cap V(G) = \{u_1\}$.
	Moreover, let $\prec'$ be the compressed extension of $\prec$ by $P$ with respect to some separator $S$ and interior $I$.
	If $S \cup I$ forms an independent set in $G$ and there is no edge $ab \in E$ with $a \in S$ or $a \preceq' u_1$ and $u_{p} \prec' b$, then~$\prec'$ is a one-page queue layout of $G + P$.
\end{observation}

Next, we want to show that compressed extensions can be used to extend a partial layout.
Before we can do that, we first discuss how to handle the case where $\Size{V(P)} = p$ is odd and $p \geq 3$ (the case $p = 1$ will never occur in our constructions by \Cref{lem:outerpath-deg-3-supergraph}); see also \Cref{fig:compressed-extension-odd}.
\begin{definition}
	\label{def:o-extension}
	Let $\prec$ be a one-page queue layout of some graph $G$ and let $P = (u_1, \ldots, u_{p})$, $p \geq 3$, be a path with an odd number of vertices such that $V(G) \cap V(P) = \{u_1\}$.
	Let $S$ and $I$ be defined as in \Cref{def:compressed-extension}.
	Moreover, let $\prec'$ be the compressed extension of $\prec$ by $P' = (u_1, \ldots, u_{p - 1})$ with separator $S$ and interior $I$.
	The \emph{\OLExtension-extension} of $\prec$ by $P$ with separator $S$ and interior $I$ is the layout $\prec''$ obtained by inserting $u_{p}$ directly right of $u_{p - 2}$ in $\prec'$, i.e., $u_{p - 2} \prec'' u_{p} \prec'' x$ for every $x \in I$.
	The \emph{\ORExtension-extension} of $\prec$ by $P$ with separator $S$ and interior $I$ is the layout $\prec''$ obtained by inserting $u_{p}$ directly right of $u_{p - 1}$ in $\prec'$, i.e., $u_{p - 1} \prec'' u_{p} \prec'' x$ for every $x \in V(G) \setminus (S \cup I)$.
\end{definition}
In the following, we use the term \emph{\EExtension-extension} to refer to a compressed extension.
We are now ready to show useful properties about the above-introduced extensions; see again \Cref{fig:compressed-extension-odd} for a visualization of them.
\begin{figure}
	\centering
	\includegraphics[page=1]{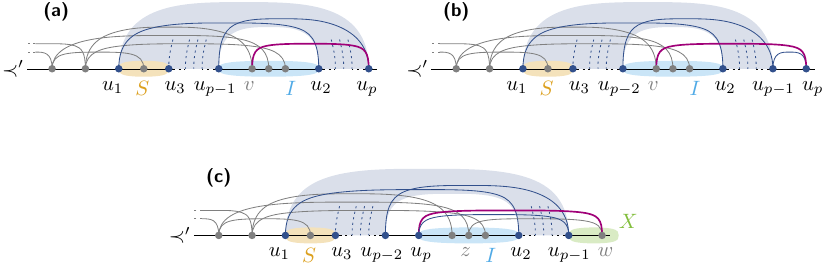}
	\caption{Visualization of the extensions $\prec'$ constructed in \Cref{lem:compressed-extension}.
		\textbf{\textsf{(a)}} The \EExtension-extension for $p$ even and $X = \emptyset$.
		\textbf{\textsf{(b)}} The \ORExtension-extension for $p$ odd and $X = \emptyset$.
		\textbf{\textsf{(c)}} The \OLExtension-extension for $p$ odd and $X = \{w\}$.	
		We indicate the graph $G$ in gray, the compressed extension from \Cref{def:compressed-extension} with the blue band, and the edge of the frontier in purple.}
	\label{fig:compressed-extension-odd}
\end{figure}
\begin{lemma}
	\label{lem:compressed-extension}
	Let $\prec$ be a one-page queue layout of $G$, and let $P = (u_1, \ldots, u_{p})$, $p \geq 2$, be a path such that $V(G) \cap V(P) = \{u_1\}$.
	Moreover, let $S$ and $I$ be as in \Cref{def:compressed-extension} such that $S \cup I$ forms an independent set in $G$, $I \neq \emptyset$, and let $X \coloneqq \{x \in V(G) \mid v \prec x\ \text{for every}\ v \in S \cup I\}$.
	Let $v \in I$ be the leftmost vertex in the interior $I$.
	The following holds.
	\begin{itemize}
		\item If $p$ is even and $X = \emptyset$, let $\prec'$ be the \EExtension-extension of $\prec$ by $P$ with separator $S$ and interior~$I$.
		Then, $\prec'$ is a one-page queue layout of $G' = G + P + \{vu_p\}$ with frontier $(v, u_p)$.
		\item If $p$ is odd and $X = \emptyset$, let $\prec'$ be the \ORExtension-extension of $\prec$ by $P$ with separator $S$ and interior~$I$.
		Then, $\prec'$ is a one-page queue layout of $G' = G + P + \{vu_p\}$ with frontier $(v, u_p)$.
		\item If $p$ is odd, $X = \{w\}$, and $N_G(w) = \{z\}$ with $z \in I$, let $\prec'$ be the \OLExtension-extension of $\prec$ by $P$ with separator~$S$ and interior~$I$.
		Then, $\prec'$ is a one-page queue layout of $G' = G + P + \{u_pw\}$ with frontier $(u_p, w)$.
	\end{itemize}
\end{lemma}
\begin{proof}
	Let $G$, $\prec$, $\prec'$, and $G'$ be as in the statement. 
	First, we discuss the \EExtension-extension case.
	We now show that~$\prec'$ is a valid one-page queue layout of $G'$ and argue afterwards that $(v, u_p)$ is a frontier of it.
	In the end, we discuss how to adapt the proof for \ORExtension- and \OLExtension-extensions.
	
	\proofsubparagraph*{Queue Layout.}
	Towards a contradiction, assume that $\prec'$ is not a one-page queue layout of $G'$.
	Hence, there are two edges $ab, cd \in E(G')$ such that $a \prec' c \prec' d \prec' b$.
	Since $\prec$ is a valid one-page queue layout of $G$ and $\prec$ occurs as sub-order in $\prec'$, we have $ab \in E(G') \setminus E(G)$ or $cd \in E(G') \setminus E(G)$ (or both).
	Since we use the compressed order for $P$, we know that not both can belong to $P$, i.e., we cannot have $ab, cd \in E(P)$; recall also \Cref{obs:compressed-extension}.
	Thus, either one edge must be from $G$, or one edge must be the edge $vu_p$.
	For the former, since $S \cup I$ is (and remains) an independent set and $X = \emptyset$, we observe that every edge $xy \in E(G)$ has one endpoint, say $x$, at $u_1$ or left of $u_1$, i.e., $x \preceq' u_1$ (recall also \Cref{obs:compressed-extension} and \Cref{fig:compressed-extension-odd}a).
	Thus, every vertex $y \in V(G)$ to the right of $u_1$ is in $S$ or $I$, i.e., $y \in S \cup I$, which is spanned over by the path $P$ in $\prec'$.
	Hence, no edge $xy \in E(G)$ can be involved in a nesting. 
	For the latter, observe that only $ab = vu_p$ is possible.
	In particular, $cd = vu_p$ is impossible, since $d \prec' b$ and $u_p$ is the rightmost vertex in $\prec'$.
	However, also the edge $vu_p$ cannot be involved in a nesting, since every vertex $c$ with $v \prec' c \prec' u_p$ is either from $I$ or from the even interval of $P$.
	In particular, this means $d$ is either to the left of $u_1$ or from the odd interval of $P$, implying $d \prec' v$; hence, no nesting here either.
	
	Combining all, we conclude that $\prec'$ is a one-page queue layout of $G'$.
	
	\proofsubparagraph*{Frontier.}
	To see that $(v,u_p)$ is a frontier of $\prec'$, we first observe 
	that $u_p$ is only adjacent to $u_{p - 1}$ (which exists as $p \geq 2$) and to $v$. 
	Moreover, $u_p$ is the last vertex in $P$ and, consequently, the rightmost vertex in $\prec'$; recall $X = \emptyset$.
	
	Finally, let $y \in V(G')$ be a vertex such that $v \prec' y \prec' u_p$.
	There are two cases, either $y \in I$ or $y \in V(P)$ (note that $y \in S$ is impossible as every vertex of $I$, in particular $v$, is right of every vertex in $S$).
	For the former, recall that $S \cup I$ is an independent set.
	Since $X = \emptyset$ (and $\prec$ occurs as suborder in $\prec'$), we conclude that $x \prec' v$ must hold.
	For the latter, observe that $y$ is in the even interval of $P$.
	Since such a vertex is only adjacent to vertices from the odd interval, we have in this case $x \prec' v$ as well.
	We conclude that $(v, u_p)$ is a frontier of $\prec'$.
	
	\proofsubparagraph*{$\boldsymbol{\ORExtension}$-extensions.}
	We now discuss $\ORExtension$-extensions.
	First observe that $p$ odd implies $p \geq 3$.
	Hence, $P' = P - \{u_p\}$ has $p - 1 \geq 2$ vertices, which is even.
	Observe that the construction resembles an \EExtension-extension for the path $P'$.
	Let $\prec''$ be the linear order that we obtain from $\prec'$ after deleting $u_p$.
	It is, by the above arguments, a one-page queue layout of $G + P' + \{vu_{p - 1}\}$.
	Since queue layouts are closed under edge deletion, it is in particular a one-page queue layout of $G + P'$.
	The only difference between $\prec''$ and $\prec'$ is the vertex $u_p$, which is rightmost in $\prec'$, i.e., placed after $u_{p - 1}$; recall $X = \emptyset$ and \Cref{fig:compressed-extension-odd}b.
	To see that $\prec'$ is a one-page queue layout of $G'$, observe that $u_p$ is only incident to $u_{p - 1}$ and $v$.
	The edge $u_{p - 1}u_p$ cannot nest any other edge since $u_{p - 1}$ and $u_p$ are the two rightmost vertices in $\prec'$.
	For the edge $vu_p$, observe that $(v,u_{p - 1})$ is a frontier of $\prec''$ (with respect to a one-page queue layout of $G + P' + \{vu_{p - 1}\}$).
	Hence, for every vertex $y$ with $v \prec'' y \prec'' u_{p - 1}$ and every neighbor $x$ of $y$, we have $x \prec'' v$ or $x = u_{p - 1}$.
	By construction, $x = u_{p - 1}$ is not possible (see also \Cref{fig:compressed-extension-odd}b).
	Hence, only the former is possible, which implies that $vu_p$ cannot be involved in any nesting.
	Thus, $\prec'$ is a one-page queue layout of $G'$.
	To see that $(v, u_p)$ is a frontier of $\prec'$, we once more observe that every $y$ with $v \prec'' y \prec'' u_{p - 1}$ and every neighbor $x$ of $y$ has $x \prec'' v$ or $x = u_{p - 1}$.
	Combined with the fact that $vu_p$ is an edge of $G'$ and $N_{G'}(u_p) = \{v, u_{p - 1}\}$, we conclude that $(v, u_p)$ is indeed a frontier of $\prec'$.
	
	\proofsubparagraph*{$\boldsymbol{\OLExtension}$-extensions.}
	Finally, we discuss $\OLExtension$-extensions.
	As before, we first observe that $p$ odd implies $p \geq 3$ and $P' = P - \{u_p\}$ has $p - 1 \geq 2$ vertices, which is even.
	The layout $\prec''$ obtained by dropping $u_p$ resembles again the layout from an \EExtension-extension with the difference that $X = \{w\}$ instead of $X = \emptyset$.
	Thus, if it contains a nesting, then it must involve an edge incident to the vertex $w$.
	However, since $N_G(w) \subseteq I$, all edges incident to $w$ span precisely the even interval of $P$.
	Since every vertex in the even interval is only incident to vertices in the odd interval (and $I$ induces an independent set in $G$), this cannot introduce a nesting.
	The same holds true for $\prec'$, since we place $u_p$ in the odd interval.
	
	To see that $(u_p, w)$ is a frontier of $\prec'$, we first observe that $w$ is the only vertex in $X$ and thus also the rightmost vertex in $\prec'$; see \Cref{fig:compressed-extension-odd}c.
	Since $N_G(w) = \{z\}$ for some $z \in I$ and we have additionally $u_pw \in E(G')$, $\deg_{G'}(w) = 2$.
	Finally, for the last criterion of frontiers, consider an arbitrary vertex $y \in V(G')$ with $u_p \prec' y \prec' w$.
	Recall that the only neighbor of $w$ other than $u_p$ is the vertex $z$.
	Every other vertex $y$ with $u_p \prec' y \prec' w$ has only neighbors $x$ which are left of $u_p$, i.e., we have $x \prec' u_p$.
	This is because $I$ is an independent set of $G$ and every vertex $y \in V(G')$ with $u_p \prec' y \prec' w$, $y \notin I$ is from the even interval of $P$, which has only neighbors in the odd interval of $P$.
	Also the neighbors of $z$ other than $w$ are to the left of $u_p$.
	Combining all, we conclude that $(u_p, w)$ is a frontier of $\prec'$.
\end{proof}

\subparagraph*{Constructing the Layout.}
With \Cref{lem:compressed-extension} at hand, we are now ready to describe how to extend a one-page queue layout of $G_i$ into one for $G_{i + 1}$ for $1 \leq i \leq f(G) - 2$; we will construct the layout for $G_1$ and $G_{f(G)} = G$ in the end.
To this end, let $c_i = u_iw_i$ and $c_{i + 1} = u_{i+1}w_{i+1}$ be the chords separating $F_{i}$ and $F_{i + 1}$, and $F_{i + 1}$ and $F_{i + 2}$, respectively.
Consider the two vertex-disjoint paths $U_i$ and $W_i$ connecting the respective endpoints of the chords.
Since $G$ is maximal, we have $\Size{V(U_i)} = 2$ or $\Size{V(W_i)} = 2$.
Depending on which of the two paths consists only of two vertices, we devise a construction rule that extends $\prec_{i}$ to $\prec_{i + 1}$ while ensuring that the endpoints of $c_{i + 1}$ constitute a frontier of $\prec_{i + 1}$; see \Cref{fig:deg-three-layout-extension} for a visualization of the construction.

\begin{lemma}
	\label{lem:outerpath-deg-3-maintain-frontier}
	Let $\prec_i$, $1 \leq i \leq f(G) - 2$, be a one-page queue layout of $G_i$ with frontier $(u_i, w_i)$.
	We can extend $\prec_{i}$ to a one-page queue layout of $G_{i + 1}$ with frontier $(u_{i + 1}, w_{i + 1})$.
\end{lemma}
\begin{proof}
	Consider the face $F_{i + 1}$ of $G$ and let $U_i = (u_i = u_i^1, \ldots, u_i^p = u_{i + 1})$ and $W_i = (w_i = w_i^1, \ldots, w_i^q = w_{i + 1})$ be the two vertex-disjoint paths on the Hamiltonian cycle $H$ connecting the respective endpoints of the chords; see \Cref{fig:deg-three-layout-extension}a.
	Since the chords form a matching, we have $u_i \neq u_{i + 1}$ and $w_i \neq w_{i + 1}$, i.e., $p, q \geq 2$.
	Moreover, since $G$ is maximal, $p = 2$ or $q = 2$ holds.
	Our idea is to obtain $\prec_{i + 1}$ by applying \Cref{lem:compressed-extension}.
	However, depending on $p > 2$ or $q > 2$ and their parity, we have to carefully select the separator $S$ and interior $I$; see \Cref{fig:deg-three-layout-extension} for a visualization.
	Before discussing the different cases, we first recall that the vertex $w_{i}$ has precisely one neighbor $z \in N(w_{i}) \setminus \{u_i\}$.
	Consider the vertex $u_i^2$, which exists as $p \geq 2$.
	We now describe safe placements for $u_i^2$ depending on the position of $z$.
	If $z \prec_i u_i$, we can place $u_i^2$ directly to the right of $w_i$.
	Since this only introduces the edge $u_iu_i^2$ and every vertex $y \in V_i$ with $u_i \prec_i y \prec_i w_i$ has only neighbors to the left of $u_i$ (since we assume $z \prec_i u_i$), this does not introduce a nesting.
	If $u_i \prec_i z$, we can place $u_i^2$ directly left of $w_i$.
	Consider the vertices $y \in V_i$ with $u_i \prec_i y \prec_i u_i^2$, which are spanned by the edge $u_iu_i^2$.
	Since $(u_i,w_i)$ is a frontier, every such vertex $y$ is either adjacent to $w_i$ (which is to the right of $u_i^2$) or to vertices $x$ with $x \prec_i u_i$.
	Thus, also here we do not introduce a nesting.
	In the following, we call this a Type-1 and Type-2 placement of $u_i^2$, respectively.
	We are now ready to discuss the individual cases.

	\proofsubparagraph*{$\boldsymbol{p = 2}$, $\boldsymbol{q}$ even:}
	For $p = 2$ and $q$ even, we first consider the case $z \prec_i u_i$.
	There, we first obtain $\prec_{i + 1}'$ via a Type-1 placement of $u_i^2 = u_{i + 1}$.
	Afterwards, $\prec_{i + 1}$ is the \EExtension-extension of $\prec_{i + 1}'$ by $W_i$ with empty separator and interior $(u_{i + 1})$; see \Cref{fig:deg-three-layout-extension}b for an example.
	Since $u_{i+1}$ is the rightmost vertex in $\prec_{i + 1}'$, this yields a valid one-page queue layout of $G_{i + 1}$ with frontier $(u_{i + 1}, w_{i + 1})$ by \Cref{lem:compressed-extension}.
	
	For the case $u_i \prec_i z$, we differentiate $q = 2$ and $q > 2$.
	In both cases, we perform a Type-2 placement of $u_i^2 = u_{i+1}$ to obtain $\prec_{i + 1}'$.
	
	For $q = 2$, we obtain $\prec_{i + 1}$ by extending $\prec_{i + 1}'$ with $w_i \prec_{i + 1} w_{i + 1}$; see \Cref{fig:deg-three-layout-extension}c.
	Since this only adds two further edges, $u_{i + 1}w_{i + 1}$ and $w_iw_{i + 1}$, and the involved three vertices are the three rightmost vertices in $\prec_{i + 1}$, this is a one-page queue layout of $G_{i + 1}$.
	Moreover, $w_i$ is the only vertex between $u_{i + 1}$ and $w_{i + 1}$, and it is not incident to $u_{i + 1}$, thus $(u_{i + 1}, w_{i + 1})$ is a frontier of $\prec_{i + 1}$.
	For $q > 2$, we place the next vertex of the path $W_i$, the vertex $w_i^2$, directly left of $u_{i+1}$.
	Since $w_i^2$ is (for now) only incident to $w_i^1 = w_i$, and the edge $w_iw_i^2$ spans precisely the vertex $u_{i + 1}$, this results in a one-page queue layout.
	The layout $\prec_{i + 1}$ is then an \ORExtension-extension of $\prec_{i + 1}'$ by the path $W_i' = (w_i^2, w_i^3, \ldots, w_i^{q} = w_{i + 1})$ with empty separator and interior $I' = (u_{i+1}, w_i)$; see \Cref{fig:deg-three-layout-extension}d.
	Since $w_i$ is the rightmost vertex in $\prec_{i + 1}'$ and $W_i'$ has an odd number of vertices, \Cref{lem:compressed-extension} ensures that $\prec_{i + 1}$ is a one-page queue layout of $G_{i + 1}$ with frontier $(u_{i + 1}, w_{i + 1})$.

	\proofsubparagraph*{$\boldsymbol{p = 2}$, $\boldsymbol{q}$ odd:}
	For $p = 2$ and $q$ odd, we perform similarly to before.
	First, we consider the case $z \prec_i u_i$ and obtain $\prec_{i + 1}'$ by a Type-1 placement of $u_i^2 = u_{i + 1}$.
	Afterwards, to obtain $\prec_{i + 1}$, we use an \ORExtension-extension of $\prec_{i + 1}'$ by $W_i$ with empty separator and interior $I = (u_{i+1})$; see \Cref{fig:deg-three-layout-extension}e.
	As $u_{i+1}$ is the rightmost vertex of $\prec_{i + 1}'$, this yields a one-page queue layout of $G_{i + 1}$ with frontier $(u_{i + 1}, w_{i + 1})$; recall \Cref{lem:compressed-extension}.
	
	Second, we consider the case $u_i \prec_i z$ and first obtain $\prec_{i + 1}'$ by a Type-2 placement of $u_i^2 = u_{i + 1}$ followed by a placement of $w_i^2$ directly left of $u_{i+1}$.
	Since $q$ odd implies $q \geq 3$, we can get $\prec_{i + 1}$ via an \EExtension-extension of $\prec_{i + 1}'$ by $W_i' = (w_i^2, \ldots, w_i^{q} = w_{i + 1})$ with empty separator and interior $I = (u_{i+1}, w_i)$; see \Cref{fig:deg-three-layout-extension}f.
	\Cref{lem:compressed-extension} ensures that it is a one-page queue layout of $G_{i + 1}$ with frontier $(u_{i + 1}, w_{i + 1})$.

	\proofsubparagraph*{$\boldsymbol{p}$ even, $\boldsymbol{q = 2}$:}
	Intuitively, we now switch the roles of $p$ and $q$ (and $U_i$ and $W_i$) in the constructions from the previous case.
	
	To be more precise, we first assume $z \prec_i u_i$ and let $\prec_{i + 1}'$ be the \EExtension-extension of $\prec_i$ by $U_i$ where the interval $\{y \in V_i \mid u_i \prec_i y \prec_i w_i\}$ forms the separator $S$ and $I = (w_i)$ is the interior; see \Cref{fig:deg-three-layout-extension}g.
	Observe that $(u_i, w_i)$ being a frontier of $\prec_i$, implies that every vertex $y \in S$ has only neighbors $x$ with $x \prec_i u_i$ or $x = w_i$.
	Since we assume $z \prec_i u_i$, we have $x \neq w_i$ and, consequently, $x \prec_i u_i$ for every neighbor $x$ of $y \in S$.
	Thus, $S \cup I$ is an independent set and we can apply \Cref{lem:compressed-extension} to conclude that $\prec_{i + 1}'$ is a one-page queue layout of $G_i + U_i$ (since queue layouts are closed under edge deletion).
	What is missing from a layout for $G_{i + 1}$ is the vertex $w_{i+1}$, which we place at the very end of the layout; recall \Cref{fig:deg-three-layout-extension}g.
	Note that the edge $w_iw_{i + 1}$ spans precisely the even interval of $U_i$.
	Since these vertices are only adjacent to vertices in the odd interval of $U_i$, the edge cannot be involved in any nesting.
	The edge $u_{i + 1}w_{i + 1}$ is between the two rightmost vertices in $\prec_{i + 1}$ and, consequently, cannot be involved in any nesting either.
	The placement of $u_{i + 1}$ and $w_{i + 1}$ allows us also to conclude that $(u_{i + 1}, w_{i + 1})$ is a frontier of $\prec_{i + 1}$.
	
	For the case $u_i \prec_i z$, we obtain $\prec_{i + 1}'$ via a Type-2 placement of $u_i^2$.
	If $p = 2$, we have $u_i^2 = u_{i + 1}$ and place $w_{i + 1}$ at the end of $\prec_{i + 1}'$ to obtain $\prec_{i + 1}$.
	This introduces two new edges, $u_{i + 1}w_{i + 1}$ and $w_iw_{i+1}$.
	The former spans precisely the vertex $w_i$, the latter spans no vertex at all.
	Thus, $\prec_{i + 1}$ is a one-page queue layout of $G_{i + 1}$ with frontier $(u_{i+1}, w_{i + 1})$ (as $w_{i+1}$ is the rightmost vertex and has degree 2 in $G_{i + 1}$).
	
	If $p > 2$, we observe that the path $U_i' = (u_i^2, \ldots, u_i^p = u_{i + 1})$ has $\Size{V(U_i')} = p - 1 \geq 3$ vertices, which is an odd number as $p$ is even.
	We place the vertex $w_{i + 1}$ to the very end of $\prec_{i + 1}'$ to obtain $\prec_{i + 1}''$.
	Afterwards, we obtain $\prec_{i + 1}$ by an \OLExtension-extension of $\prec_{i + 1}''$ by $U_i'$ with empty separator and interior $w_i$; see \Cref{fig:deg-three-layout-extension}h.
	By \Cref{lem:compressed-extension}, $\prec_{i + 1}$ is a one-page queue layout of $G_{i + 1}$; observe that $X = \{w_{i + 1}\}$ holds.
	
	\proofsubparagraph*{$\boldsymbol{p}$ odd, $\boldsymbol{q = 2}$:}
	Finally, we consider the last case with $p$ odd and $q = 2$.
	
	First, assume $z \prec_i u_i$.
	In this case, let $\prec_{i + 1}'$ be obtained by placing $w_{i + 1}$ directly right of~$w_{i}$.
	As this only introduces the edge $w_iw_{i+1}$ between the two rightmost vertices in $\prec_{i + 1}'$, the layout is a one-page queue layout.
	Afterwards, we obtain $\prec_{i + 1}$ by an \OLExtension-extension of $\prec_{i + 1}'$ by $U_i$ where the separator $S$ is formed by the vertices $y$ with $u_i \prec_{i} y \prec_{i} w_i$ and with interior $I = (w_i)$; see \Cref{fig:deg-three-layout-extension}i.
	Since $(u_i, w_i)$ is a frontier of $\prec_i$, every vertex $y \in S$ has only neighbors $x$ with $x \prec_i u_i$ or $x = w_i$.
	Moreover, since we assume $z \prec_i u_i$, we have $x \neq w_i$ and, consequently, $x \prec_i u_i$ for every neighbor $x$ of $y \in S$.
	Thus, $S \cup I$ is an independent set and we can apply \Cref{lem:compressed-extension} to conclude that $\prec_{i + 1}$ is a one-page queue layout of $G_{i + 1}$ with frontier $(u_{i + 1}, w_{i + 1})$; observe that $X = \{w_{i + 1}\}$ holds.
	
	Second, assume $u_i \prec_i z$.
	In this case, we first obtain $\prec_{i + 1}'$ by a Type-2 placement of $u_i^2$.
	From $\prec_{i + 1}'$, we can get a layout of $G_i + U_i$ via an \EExtension-extension by $U_i' = (u_i^2, \ldots, u_i^p = u_{i + 1})$ with empty separator and interior $I = (w_i)$.
	To get $\prec_{i + 1}$, we place $w_{i + 1}$ to the very right of the layout; see \Cref{fig:deg-three-layout-extension}j.
	To see that $\prec_{i + 1}$ is a one-page queue layout of $G_{i + 1}$, observe that we introduce two new edges: $w_{i}w_{i + 1}$ and $u_{i + 1}w_{i + 1}$.
	The former spans precisely the even interval of $U_i'$ and can thus not be involved in any nesting (since all vertices are adjacent to the odd interval or $w_{i + 1}$).
	The latter edge spans the two rightmost vertices, i.e., cannot be involved in a nesting.
	Together with the fact that $w_{i + 1}$ has two neighbors in $G_{i + 1}$ and is rightmost in $\prec_{i + 1}$, the last observation allows us to conclude that $(u_{i + 1}, w_{i + 1})$ is a frontier of $\prec_{i + 1}$.

	\proofsubparagraph*{Wrapping Up.}
	In all considered cases, we have constructed a one-page queue layout $\prec_{i + 1}$ of~$G_{i + 1}$ with frontier $(u_{i + 1}, w_{i + 1})$.
	Thus, the statement follows by combining them.
\end{proof}

\begin{figure}
	\centering
	\includegraphics[page=1]{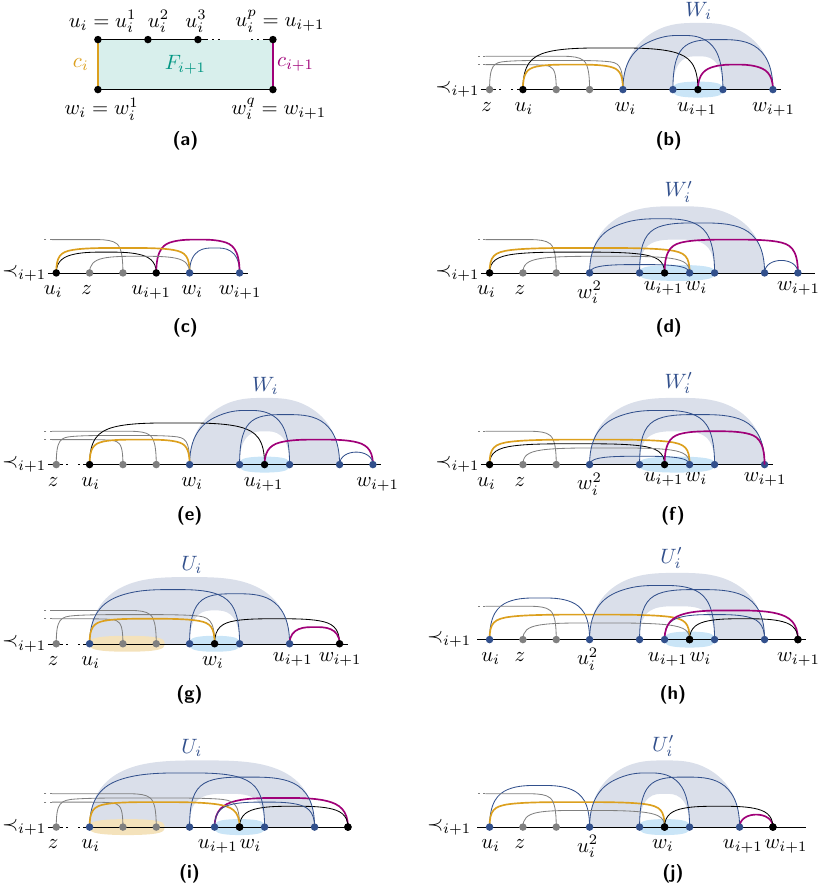}
	\caption{Visualization of the construction behind \Cref{lem:compressed-extension}.
		\textbf{\textsf{(a)}} The face $F_{i + 1}$, which is bounded by the chords $c_i = u_iw_i$ (yellow) and $c_{i+1} = u_{i + 1}w_{i+1}$ (purple), together with the two paths $U_i$ and $W_i$ consisting of $p$ and $q$ vertices, respectively.
		Here $q = 2$ and $p \geq 2$.
		The tuple $(u_i,w_i)$ is a frontier of $\prec_i$.
		\textbf{\textsf{(b)}}--\textbf{\textsf{(d)}} Case $p = 2$ and $q$ even.
		\textbf{\textsf{(e)}} and \textbf{\textsf{(f)}} Case $p = 2$ and $q$ odd.
		\textbf{\textsf{(g)}} and \textbf{\textsf{(h)}} Case $p$ even and $q = 2$.
		\textbf{\textsf{(i)}} and \textbf{\textsf{(j)}} Case $p$ odd and $q = 2$.
		We indicate the \EExtension-, \ORExtension-, and \OLExtension-extension with the blue band.
		The separator $S$ and interior $I$ are indicated in yellow and light blue, respectively.
	}
	\label{fig:deg-three-layout-extension}
\end{figure}

We now iteratively apply \Cref{lem:outerpath-deg-3-maintain-frontier} to obtain a one-page queue layout of an outerpath with maximum degree 3.
\begin{proposition}
	\label{thm:outerpath-max-deg-3}
	Every outerpath with vertex degree at most~$3$ has queue number~$1$.
\end{proposition}
\begin{proof}
	Let $G$ be an outerpath with maximum degree 3.
	We assume that $G$ has at least four vertices; otherwise we can take any vertex order of $G$.
	
	By \Cref{lem:outerpath-deg-3-supergraph}, we let $G^+$ be a maximal outerpath of degree 3 which is a supergraph of $G$.
	Let $v_1$ and $v_n$ be the unique vertices of degree two in the first and last face of $G^+$, respectively.
	Let $F_1$ be the leftmost face (in a drawing) of $G^+$.
	By \Cref{lem:outerpath-deg-3-supergraph}, $F_1$ is a triangle.
	Let it be composed of the vertices $v_1$, $u_1$, and $w_1$, where $u_1w_1$ forms the chord separating $F_1$ from the next face $F_2$.
	
	We now construct a linear order $\prec_1$ of $v_1$, $u_1$, and $w_1$ as $v_1 \prec_1 u_1 \prec_1 w_1$.
	Observe that $\prec_1$ is a one-page queue layout of $G_1$ with frontier $(u_1, w_1)$.
	Let $k$ be the number of faces of $G^+$ and let $F_k$ be composed of the vertices $u_{k - 1}$, $w_{k - 1}$, and $v_{n}$, where $c_{k - 1} = u_{k - 1}w_{k - 1}$; recall \Cref{lem:outerpath-deg-3-supergraph}.
	By iteratively applying \Cref{lem:outerpath-deg-3-maintain-frontier}, we can obtain a one-page queue layout $\prec_{k - 1}$ of $G_{k - 1}$ with frontier $(u_{k - 1}, w_{k - 1})$.
	
	What is missing from a one-page queue layout of $G^+$ is the vertex $v_n$ and its two incident edges $u_{k - 1}v_n$ and $w_{k - 1}v_n$.
	Since $(u_{k - 1}, w_{k - 1})$ is a frontier of $\prec_{k - 1}$, $w_{k - 1}$ is the rightmost vertex in $\prec_{k - 1}$; recall \Cref{def:outerpath-deg-3-frontier}.
	Moreover, it has, apart from $u_{k - 1}$, one other neighbor $z \in V_{k - 1}$.
	
	If $z \prec_{k - 1} u_{k - 1}$, we extend $\prec_{k - 1}$ by placing $v_{n}$ as the new rightmost vertex, obtaining $\prec$.	
	To see that $\prec$ is a one-page queue layout, observe that only the edges $u_{k - 1}v_n$ and $w_{k - 1}v_n$ have been added.
	Thus, only they can be involved in any nesting (since $\prec_{k - 1}$ is a one-page queue layout of $G_{k - 1}$).
	Since $w_{k - 1}$ is directly to the left of $v_n$, $w_{k - 1}v_n$ cannot nest any edge (since the only edge above it is $u_{k - 1}v_{n}$).
	Regarding the edge $u_{k - 1}v_{n}$, the frontier (together with the assumption $z \prec_{k - 1} u_{k - 1}$) ensures that every vertex $y$ with $u_{k - 1} \prec y \prec w_{k - 1} \prec v_n$ is only adjacent to vertices $x \prec u_{k - 1}$.
	Thus, also the edge $u_{k - 1}v_{n}$ cannot nest any other edge.
	
	If $u_{k - 1} \prec_{k - 1} z$, we extend $\prec_{k - 1}$ by placing $v_{n}$ directly left of $w_{k - 1}$, obtaining $\prec$.	
	Again, only the edges $u_{k - 1}v_n$ and $w_{k - 1}v_n$ have been added, and only they can be involved in any nesting (since $\prec_{k - 1}$ is a one-page queue layout of $G_{k - 1}$).
	Since $v_n$ is directly next to $w_{k - 1}$, the edge $w_{k - 1}v_n$ cannot nest any other edge (the only edge above it is $u_{k - 1}w_{k - 1}$).
	Regarding the edge $u_{k - 1}v_{n}$, every vertex $y$ with $u_{k - 1} \prec y \prec v_n$ is either adjacent to vertices $x \prec u_{k - 1}$ or to $w_{k - 1}$.
	In neither case can the edge $u_{k - 1}v_{n}$ be involved in a nesting.

	Combining all, we conclude that $\prec$ is a one-page queue layout of $G^+$.
	Since queue layouts are closed under vertex removal, this also shows that $G$ has a one-page queue layout.
\end{proof}

\subsection{Outerpaths with Maximum Degree~4}
\label{sub:degree-4-outerpaths}
In this subsection, we prove that there are outerpaths of maximum degree~$4$ that do not admit a $1$-queue layout \lQL.

\begin{proposition}
\label{prop:max-deg-4-outerpaths}
There are outerpaths of maximum degree~4 that do not admit a $1$-queue layout.
\end{proposition}

\begin{figure}[h!]
    \centering
    \includegraphics[width=0.4 \textwidth]{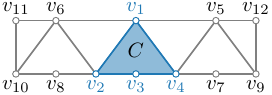}
    \caption{An outerpath of maximum degree~$4$ that does not admit a $1$-queue layout. }
    \label{fig:non-max-deg-4-counterexample}
\end{figure}

\begin{proof}
Consider the outerpath that is shown in \Cref{fig:non-max-deg-4-counterexample} and observe that its maximum degree is $4$. Suppose for the sake of contradiction that it admits a 1-queue layout $L$. We will show that it contains a 2-rainbow in $L$. Consider the 4-cycle $C$, highlighted in blue in \Cref{fig:non-max-deg-4-counterexample}. W.l.o.g.\ (by reversing $L$ if necessary) we assume that $v_2\prec v_4$ in $L$.
Since $L$ is a $1$-queue layout one of the following holds: 
\begin{enumerate*}[label={\textcolor{lipicsGray}{\bf\textsf{(S.\arabic*)}}},ref={\textsf{S.\arabic*}}]
\item\label{s:1} $v_2 \prec v_1 \prec v_3 \prec v_4$
\item\label{s:2} $v_2 \prec v_3 \prec v_1 \prec v_4$
\item\label{s:3} $v_1 \prec v_2 \prec v_4 \prec v_3$
\item\label{s:4} $v_3 \prec v_2 \prec v_4 \prec v_1$. %
\end{enumerate*}
Indeed, all other permutations yield a nesting in $L$:
if $v_2 \prec \{v_1,v_4\} \prec v_3$ or $v_1 \prec \{v_2,v_3\} \prec v_4$, then edges $(v_2,v_3)$ and $(v_1,v_4)$ form a $2$-rainbow;
if $v_2 \prec \{v_3,v_4\} \prec v_1$ or $v_3 \prec \{v_1,v_2\} \prec v_4$, then edges $(v_2,v_1)$ and $(v_4,v_3)$ form a $2$-rainbow.
We focus on Cases S.1 and S.2, which are more involved. Afterwards, we present Cases S.3 and S.4

\medskip
\noindent\textbf{Case \ref{s:1}:} $v_2 \prec v_1 \prec v_3 \prec v_4$.

\begin{figure}[h!]
    \centering
    \includegraphics[page=1]{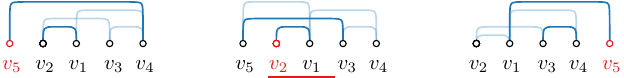}
\end{figure}

We next claim that vertex $v_5$ appears between $v_3$ and $v_4$ in $L$, that is $v_3 \prec v_5 \prec v_4$. To see this, assume to the contrary that $v_5$ is not between them. Then, for each of the remaining cases, it is not difficult to see that a $2$-rainbow is formed in $L$:
if $v_5 \prec v_2$, then edges $(v_5,v_4)$ and $(v_2,v_3)$ form a 2-rainbow;
if $v_2 \prec v_5 \prec v_3$, then edges $(v_2,v_3)$ and $(v_5,v_1)$ form a 2-rainbow; 
if $v_4 \prec v_5$, then edges $ (v_1,v_5)$ and $(v_3,v_4)$ form a 2-rainbow.
Hence, vertex $v_5$ appears between $v_3$ and $v_4$ in $L$, as we claimed. Next, consider vertex $v_6$. 

\begin{figure}[h!]
    \centering
    \includegraphics[page=2]{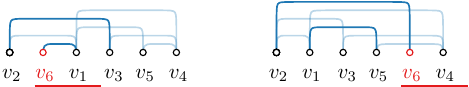}
\end{figure}

\noindent If $v_2 \prec v_6 \prec v_3$, then $(v_2,v_3)$ nests $(v_6,v_1)$; if $v_5 \prec v_6$, then $(v_6,v_2)$ nests $(v_1,v_5)$. Thus, two cases remain: either $v_6 \prec v_2$ or $v_3 \prec v_6 \prec v_5$, which we analyze separately below.

\medskip
\noindent\textbf{Subcase 1.1:} $v_6 \prec v_2$.
\begin{figure}[h!]
    \centering
    \includegraphics[page=3]{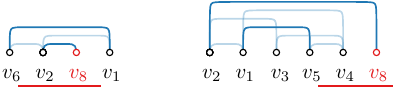}
\end{figure}
We proceed with vertex $v_8$. There exist three cases of vertex $v_8$ in $L$: \emph{(a)} $v_8 \prec v_6$, \emph{(b)} $v_1 \prec v_8 \prec v_3$, or \emph{(c)} $v_3 \prec v_8 \prec v_5$. The remaining cases do not hold as if $v_6 \prec v_8 \prec v_1$, then $(v_6,v_1)$ nests $(v_8,v_2)$ and if $v_5 \prec v_8$, then $(v_8,v_2)$ nests $(v_1,v_5)$. 
In the following, we show that each of the three cases yields a $2$-rainbow, which is a contradiction to the fact that $L$ is a $1$-queue layout. We begin with cases \emph{(b)} and \emph{(c)}.

\smallskip\noindent\emph{(b)} Recall that in this case the order of vertices $v_6,v_8,v_2,v_1,v_4,v_3,v_5$ is as follows: $v_6 \prec v_2 \prec v_1 \prec v_8 \prec v_3 \prec v_5 \prec v_4$. We show that vertex $v_{10}$ results in a $2$-rainbow in $L$:

\begin{figure}[h!]
    \centering
    \includegraphics[page=4]{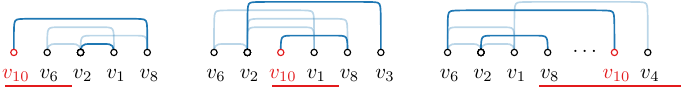}
\end{figure}

\noindent if $v_{10} \prec v_2$, then $(v_{10},v_8)$ nests $(v_2,v_1)$;
if $v_2 \prec v_{10} \prec v_8$, then $(v_2,v_3)$ nests $(v_{10},v_8)$;
if $v_8 \prec v_{10}$, then $(v_{10},v_6)$ nests $(v_8,v_2)$.
Hence, vertex $v_{10}$ cannot appear in $L$ without creating a $2$-rainbow.

\smallskip\noindent\emph{(c)} In this case the order of the vertices is $v_6 \prec v_2 \prec v_1 \prec v_3 \prec v_8 \prec v_5 \prec v_4$. We claim that vertex $v_{10}$ can only be between $v_2$ and $v_1$ that is $ v_2 \prec v_{10} \prec v_1$. 

\begin{figure}[h!]
    \centering
    \includegraphics[page=5]{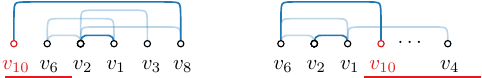}
\end{figure}

\noindent Indeed, if $v_{10} \prec v_2$, then $(v_{10},v_8)$ nests $(v_2,v_1)$;
if $v_1 \prec v_{10}$, then $(v_{10},v_6)$ nests $(v_2,v_1)$.
Hence, vertex $v_{10}$ appears between $v_2$ and $v_1$ in $L$, as we claimed. Next, we proceed with vertex $v_{11}$.

\begin{figure}[h!]
    \centering
    \includegraphics[page=6]{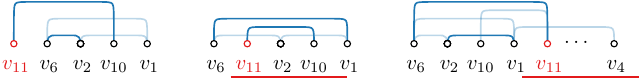}
\end{figure}

\noindent In each of the remaining cases, a $2$-rainbow is formed in $L$:
if $v_{11} \prec v_6$, then $(v_{11},v_{10})$ nests $(v_6,v_2)$;
if $v_6 \prec v_{11} \prec v_1$, then $(v_6,v_1)$ nests $(v_{11},v_{10})$;
if $v_1 \prec v_{11}$, then $(v_{11},v_6)$ nests $(v_2,v_1)$.
Hence, vertex $v_{11}$ cannot appear in $L$ without creating a $2$-rainbow.

\smallskip\noindent\emph{(a)} We first consider vertex $v_{10}$ and we will show that it appears before $v_8$ in $L$:

\begin{figure}[h!]
    \centering
    \includegraphics[page=7]{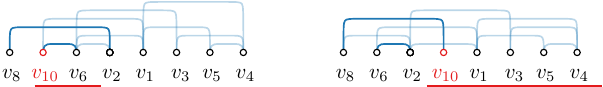}
\end{figure}

\noindent if $v_8 \prec v_{10} \prec v_2$, then $(v_8,v_2)$ nests $(v_{10},v_6)$;
if $v_2 \prec v_{10} $, then $(v_{10},v_8)$ nests $(v_6,v_2)$.
Hence, the only valid option is $v_{10} \prec v_8$, as we claim.
Next, we show that vertex $v_{11}$ is between vertices $v_{10}$ and $v_8$. Indeed:

\begin{figure}[h!]
    \centering
    \includegraphics[page=8]{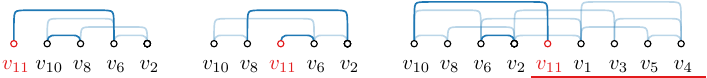}
\end{figure}

\noindent if $v_{11} \prec v_{10}$, then $(v_{11},v_6)$ nests $(v_{10},v_8)$;
if $v_8 \prec v_{11} \prec v_2$, then $(v_8,v_2)$ nests $(v_{11},v_6)$;
if $v_2 \prec v_{11}$, then $(v_{11},v_{10})$ nests $(v_6,v_2)$.
Thus, the total order is $v_{10}\prec v_{11}\prec v_8\prec v_6\prec v_2\prec v_1\prec v_3\prec v_5\prec v_4$.
Regarding vertex $v_{12}$ there exist two cases $v_2 \prec v_{12} \prec v_1$ or $v_4 \prec v_{12}$:

\begin{figure}[h!]
    \centering
    \includegraphics[page=9]{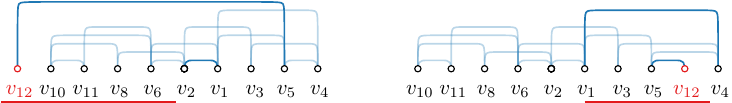}
\end{figure}

\noindent if $v_{12} \prec v_2$, then $(v_{12},v_5)$ nests $(v_2,v_1)$;
if $v_1 \prec v_{12} \prec v_4$, then $(v_1,v_4)$ nests $(v_{12},v_5)$.

First, assume that $v_2 \prec v_{12} \prec v_1$; the total order in $L$ is $v_{10} \prec v_{11} \prec v_8 \prec v_6 \prec v_2 \prec v_{12} \prec v_1 \prec v_3 \prec v_5 \prec v_4$.
We will show that vertex $v_9$ results in a $2$-rainbow in $L$:

\begin{figure}[h!]
    \centering
    \includegraphics[page=10]{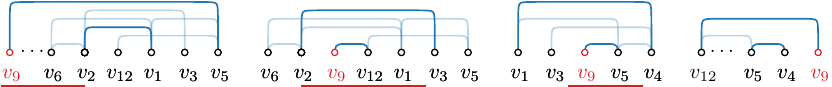}
\end{figure}

\noindent if $v_9 \prec v_2$, then $(v_9,v_5)$ nests $(v_2,v_1)$;
if $v_2 \prec v_9 \prec v_3$, then $(v_2,v_3)$ nests $(v_9,v_{12})$;
if $v_3 \prec v_9 \prec v_4$, then $(v_1,v_4)$ nests $(v_9,v_5)$;
if $v_4 \prec v_9$, then $(v_9,v_{12})$ nests $(v_5,v_4)$.

Next, assume that $v_4 \prec v_{12}$, so the total order is $v_{10} \prec v_{11} \prec v_8 \prec v_6 \prec v_2 \prec v_1 \prec v_3 \prec v_5 \prec v_4 \prec v_{12}$. 
Considering again vertex $v_9$ we will show that there exist two possible cases, either $v_4 \prec v_9 \prec v_{12}$ or $v_{12} \prec v_9$:

\begin{figure}[h!]
    \centering
    \includegraphics[page=11]{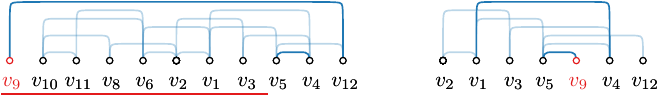}
\end{figure}

\noindent if $v_9 \prec v_5$, then $(v_9,v_{12})$ nests $(v_5,v_4)$;
if $v_5 \prec v_9 \prec v_4$, then $(v_1,v_4)$ nests $(v_5,v_9)$.

We proceed with vertex $v_7$ for both of the two previous cases. Let $x\in\{v_9,v_{12}\}$ denote whichever of the two lies further to the right. We will prove that vertex $v_7$ always yields a nesting in $L$:

\begin{figure}[h!]
    \centering
    \includegraphics[page=12]{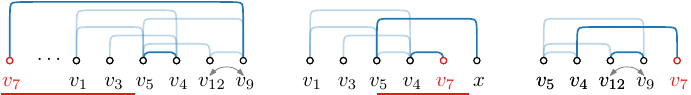}
\end{figure}

if $v_7 \prec v_5$, then $(v_7,v_9)$ nests $(v_5,v_4)$;
if $v_5\prec v_7\prec x$, then $(v_5,x)$ nests $(v_7,v_4)$.
Finally, if $x\prec v_7$, then $(v_7,v_4)$ nests $(v_9,v_{12})$.
Indeed, vertex $v_7$ results in a $2$-rainbow in $L$.

\noindent\textbf{Subcase 1.2:} $v_3 \prec v_6 \prec v_5$. First we will prove that vertex $v_{11}$ is between $v_2$ and $v_1$ in $L$ and afterwards that vertex $v_{10}$  always yields a $2$-rainbow. Indeed, for vertex $v_{11}$:

\begin{figure}[h!]
    \centering
    \includegraphics[page=13]{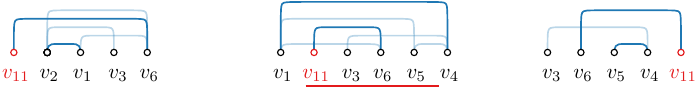}
\end{figure}

\noindent if $v_{11} \prec v_2$, then $(v_{11},v_6)$ nests $(v_2,v_1)$;
if $v_1 \prec v_{11} \prec v_4$, then $(v_1,v_4)$ nests $(v_{11},v_6)$;
if $v_4 \prec v_{11}$, then $(v_{11},v_6)$ nests $(v_5,v_4)$.
Hence, the total order is $v_2 \prec v_{11} \prec v_1 \prec v_3 \prec v_6 \prec v_5 \prec v_4$. For this order of $L$, vertex $v_{10}$ implies a nesting:

\begin{figure}[h!]
    \centering
    \includegraphics[page=14]{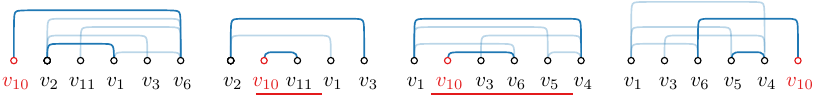}
\end{figure}

\noindent if $v_{10} \prec v_2$, then $(v_{10},v_6)$ nests $(v_2,v_1)$;
if $v_2 \prec v_{10} \prec v_1$, then $(v_2,v_3)$ nests $(v_{11},v_{10})$;
if $v_1 \prec v_{10} \prec v_4$, then $(v_1,v_4)$ nests $(v_{10},v_6)$;
if $v_4 \prec v_{10}$, then $(v_{10},v_6)$ nests $(v_5,v_4)$.

In all configurations of Case~S.1, a $2$-rainbow arises, contradicting the fact that $L$ is a $1$-queue layout.

\medskip
\noindent\textbf{Case S.2:} $v_2 \prec v_3 \prec v_1 \prec v_4$.

\begin{figure}[h!]
    \centering
    \includegraphics[page=1]{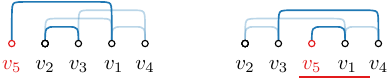}
\end{figure}

\noindent First, consider vertex $v_5$. We will show that there exist two cases: either $v_5$ is between $v_2$ and $v_3$ or it is after $v_4$:
if $v_5 \prec v_2$, then $(v_5,v_4)$ nests $(v_2,v_3)$;
if $v_3 \prec v_5 \prec v_4$, then $(v_3,v_4)$ nests $(v_5,v_1)$.
We analyze the two cases $v_2 \prec v_5 \prec v_3$ and $v_4 \prec v_5$ separately below.

\smallskip\noindent\textbf{Subcase 2.1:} $v_2 \prec v_5 \prec v_3$. Next we prove that vertex $v_6$ is between $v_2$ and $v_5$ in $L$:

\begin{figure}[h!]
    \centering
    \includegraphics[page=2]{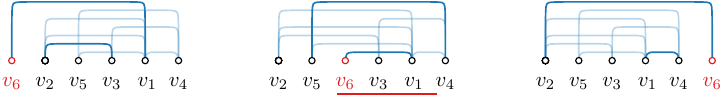}
\end{figure}

\noindent if $v_6 \prec v_2$, then $(v_6,v_1)$ nests $(v_2,v_3)$;
if $v_5 \prec v_6 \prec v_4$, then $(v_5,v_4)$ nests $(v_6,v_1)$;
if $v_4 \prec v_6$, then $(v_6,v_2)$ nests $(v_1,v_4)$.
Next, vertex $v_{11}$ must be before vertex $v_2$:

\begin{figure}[h!]
    \centering
    \includegraphics[page=3]{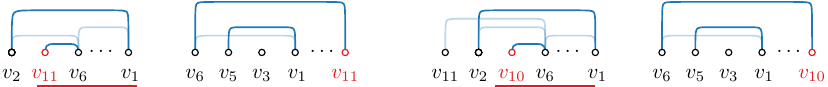}
\end{figure}

\noindent if $v_2 \prec v_{11} \prec v_1$, then $(v_2,v_1)$ nests $(v_{11},v_6)$;
if $v_1 \prec v_{11}$, then $(v_{11},v_6)$ nests $(v_5,v_1)$.
Thus, $v_{11} \prec v_2$.
Finally, for vertex $v_{10}$, we observe that:
if $v_2 \prec v_{10} \prec v_1$, then $(v_2,v_1)$ nests $(v_{10},v_6)$;
if $v_1 \prec v_{10}$, then $(v_{10},v_6)$ nests $(v_5,v_1)$.
Hence, for vertex $v_{10}$ there exist two cases: either $v_{10} \prec v_{11}$ or $v_{11} \prec v_{10} \prec v_2$ in $L$.

\begin{figure}[h!]
    \centering
    \includegraphics[page=4]{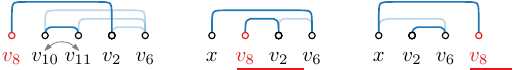}
\end{figure}

\noindent The total order of $L$ so far is $v_{11},v_{10} \prec v_2 \prec v_6 \prec v_5 \prec v_3 \prec v_1 \prec v_4$. Consider now vertex $v_8$. Let $x\in\{v_{11},v_{10}\}$ denote whichever of the two lies further to the left in $L$.  We will prove that vertex $v_8$ always yields a nesting in $L$:
if $v_8 \prec x$, then $(v_8,v_2)$ nests $(v_{10},v_{11})$;
if $x \prec v_8 \prec v_6$, then $(x,v_6)$ nests $(v_8,v_2)$.
Finally, if $v_6 \prec v_8$, then $(v_8,v_{10})$ nests $(v_2,v_6)$.
In all configurations, vertex $v_8$ forms a $2$-rainbow in $L$ leading to a contradiction. 

\medskip
\noindent\textbf{Subcase 2.2:} $v_4 \prec v_5$. For vertex $v_6$, we observe that:

\begin{figure}[h!]
    \centering
    \includegraphics[page=5]{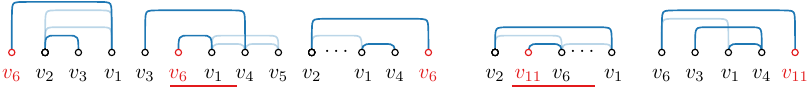}
\end{figure}

\noindent if $v_6 \prec v_2$, then $(v_6,v_1)$ nests $(v_2,v_3)$;
if $v_3 \prec v_6 \prec v_4$, then $(v_3,v_4)$ nests $(v_6,v_1)$;
if $v_4 \prec v_6$, then $(v_6,v_2)$ nests $(v_1,v_4)$.
Hence, vertex $v_6$ can only be between $v_2$ and $v_3$ in $L$ that is $v_2 \prec v_6 \prec v_3$.
Next we prove that for vertex $v_{11}$ we have two cases, either it is before vertex $v_2$ or between $v_1$ and $v_4$. Indeed,   
if $v_2 \prec v_{11} \prec v_1$, then $(v_2,v_1)$ nests $(v_{11},v_6)$;
if $v_4 \prec v_{11}$, then $(v_{11},v_6)$ nests $(v_4,v_1)$.
Thus, \emph{(a)} $v_{11} \prec v_2$ or \emph{(b)} $v_1 \prec v_{11} \prec v_4$. We analyze these two options separately below.

\begin{figure}[h!]
    \centering
    \includegraphics[page=6]{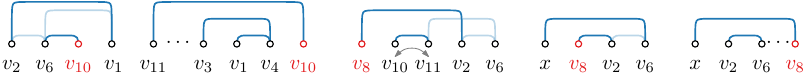}
\end{figure}

\smallskip \noindent\emph{(a)} $v_{11} \prec v_2$. For vertex $v_{10}$, we observe that:
if $v_2 \prec v_{10} \prec v_1$, then $(v_2,v_1)$ nests $(v_{10},v_6)$;
if $v_1 \prec v_{10}$, then $(v_{11},v_{10})$ nests $(v_1,v_2)$.
Hence, for vertex $v_{10}$ there are two options: either $v_{10} \prec v_{11}$ or $v_{11} \prec v_{10} \prec v_2$.
Next, we will show that vertex $v_8$ always yields a nesting in $L$ for both options of vertex $v_{10}$. Again, let $x\in\{v_{11},v_{10}\}$ denote whichever of the two lies further to the left in $L$:
if $v_8 \prec x$, then $(v_8,v_2)$ nests $(v_{11},v_{10})$;
if $x \prec v_8 \prec v_6$, then $(x,v_6)$ nests $(v_8,v_2)$;
if $v_6 \prec v_8$, then $(v_8,v_{10})$ nests $(v_2,v_6)$.
Indeed, all cases of vertex $v_8$ result in a nesting in $L$.

\begin{figure}[h!]
    \centering
    \includegraphics[page=7]{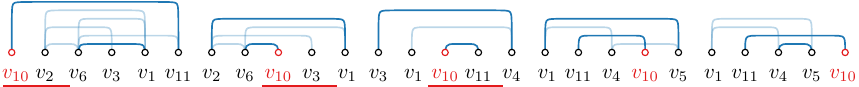}
\end{figure}

\smallskip \noindent\emph{(b)}  $v_1 \prec v_{11} \prec v_4$.
For vertex $v_{10}$ we will prove that it always forms a $2$-rainbow in $L$:
if $v_{10} \prec v_6$, then $(v_{10},v_{11})$ nests $(v_6,v_1)$;
if $v_6 \prec v_{10} \prec v_1$, then $(v_2,v_1)$ nests $(v_{10},v_6)$;
if $v_1 \prec v_{10} \prec v_4$, then $(v_3,v_4)$ nests $(v_{10},v_{11})$;
if $v_4 \prec v_{10} \prec v_5$, then $(v_1,v_5)$ nests $(v_{10},v_{11})$;
if $v_5 \prec v_{10}$, then $(v_{10},v_{11})$ nests $(v_4,v_5)$.
In all cases, a nesting occurs, yielding a contradiction.

In all configurations of Case~S.2, a $2$-rainbow arises, contradicting the fact that $L$ is a $1$-queue layout.

\noindent \textbf{Case S.3: $v_1 \prec v_2 \prec v_4 \prec v_3$.} 

\begin{figure}[h!]
    \centering
    \includegraphics[page=1]{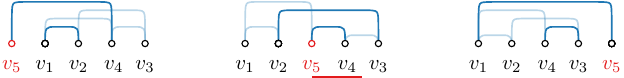}
\end{figure}

\noindent First we will show that vertex $v_5$ appears between $v_1$ and $v_2$ in $L$ and afterwards that vertex $v_6$ cannot appear in $L$ without creating a nesting.
Indeed, the following holds for vertex $v_5$:
if $v_5 \prec v_1$, then edges $(v_5,v_4)$ and $(v_1,v_2)$ form a $2$-rainbow;
if $v_2 \prec v_5 \prec v_3$, then edges $(v_2,v_3)$ and $(v_4,v_5)$ form a $2$-rainbow;
if $v_3 \prec v_5$, then edges $(v_1,v_5)$ and $(v_4,v_3)$ form a $2$-rainbow.
Thus, indeed $v_1 \prec v_5 \prec v_2$ as we claimed.

\begin{figure}[h!]
    \centering
    \includegraphics[page=2]{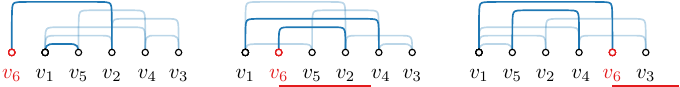}
\end{figure}

\noindent Next, consider vertex $v_6$. We will prove that it always results in a nesting in $L$:
if $v_6 \prec v_1$, then $(v_6,v_2)$ nests $(v_1,v_5)$;
if $v_1 \prec v_6 \prec v_4$, then $(v_1,v_4)$ nests $(v_6,v_2)$;
if $v_4 \prec v_6$, then $(v_6,v_1)$ nests $(v_5,v_4)$.
In all cases, a nesting occurs, yielding a contradiction.

\noindent \textbf{Case S.4: $v_3 \prec v_2 \prec v_4 \prec v_1$.} 

\begin{figure}[h!]
    \centering
    \includegraphics[page=1]{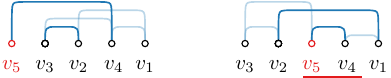}
\end{figure}

We will show that for vertex $v_5$ there exist two cases, either \emph{(a)} $v_3 \prec v_5 \prec v_2$ or \emph{(b)} $v_1 \prec v_5$.
If $v_5 \prec v_3$, then edge $(v_4,v_5)$ nests $(v_2,v_3)$;
if $v_2 \prec v_5 \prec v_1$, then edge $(v_2,v_1)$ nests $(v_4,v_5)$.
Therefore, vertex $v_5$ appears either between $v_3$ and $v_2$ or after vertex $v_1$.
For both cases of vertex $v_5$, vertex $v_6$ forms a $2$-rainbow in $L$.

\begin{figure}[h!]
    \centering
    \includegraphics[page=2]{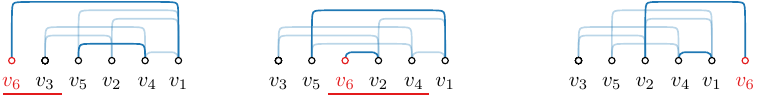}
\end{figure}

\noindent \emph{(a)} $v_3 \prec v_5 \prec v_2$:
if $v_6 \prec v_5$, then edges $(v_6,v_1)$ and $(v_5,v_4)$ form a $2$-rainbow;
if $v_5 \prec v_6 \prec v_1$, then edges $(v_6,v_2)$ and $(v_1,v_5)$ form a $2$-rainbow;
if $v_1 \prec v_6$, then edges $(v_6,v_2)$ and $(v_1,v_4)$ form a $2$-rainbow.
Therefore, vertex $v_6$ cannot appear in $L$, as it always yields a nesting.

\begin{figure}[h!]
    \centering
    \includegraphics[page=3]{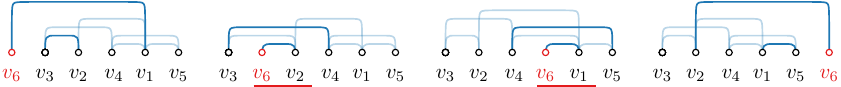}
\end{figure}

\noindent \emph{(b)} $v_1 \prec v_5 $. Similarly to \emph{(a)}, vertex $v_6$ always leads to a contradiction.
Indeed, 
if $v_6 \prec v_3$, then edges $(v_6,v_1)$ and $(v_2,v_3)$ form a $2$-rainbow;
if $ v_3 \prec v_6 \prec v_4$, then edges $(v_6,v_2)$ and $(v_4,v_3)$ form a $2$-rainbow;
if $v_4 \prec v_6 \prec v_5$, then edges $(v_6,v_1)$ and $(v_4,v_5)$ form a $2$-rainbow;
if $v_5 \prec v_6$, then edges $(v_6,v_2)$ and $(v_1,v_4)$ form a $2$-rainbow.

\medskip
Thus, in every case, any linear order contains a $2$-rainbow. Therefore, the graph of~\cref{fig:non-max-deg-4-counterexample} does not admit a $1$-queue layout.
\end{proof}

In contrast to \cref{prop:degree-4-max-outerpath}, adding the edges $(v_2,v_4)$, $(v_6,v_8)$, and $(v_9,v_{10})$ to the graph depicted in \cref{fig:non-max-deg-4-counterexample} gives a maximal outerpath of maximum degree~5 that does not admit a $1$-queue layout.

\begin{observation}
    \label{obs:max-outerpath-max-deg-5}
    There are maximal outerpaths of maximum degree~5 that do not admit a $1$-queue layout.
\end{observation}

\section{Concluding Remarks}

We have shown that deciding whether an outerplanar graph has queue number~1 is NP-hard. In contrast, we have shown that for maximal outerplanar graphs, the problem can be solved in linear time. A natural graph class to consider next are biconnected outerplanar graphs, as they generalize maximal outerplanar graphs while being more restricted than general outerplanar graphs.

\begin{open}
What is the complexity of deciding whether a biconnected outerplanar graph has queue number 1?
\end{open}

Outerpaths are an important ingredient in our approach to recognizing maximal outerplanar graphs of queue number 1 and also of independent interest. In addition to several existential results for outerpaths, we have given a linear-time algorithm to decide whether maximal outerpaths have queue number 1. It remains open whether maximality is necessary for efficient recognition.

\begin{open}
Can it be decided in polynomial time whether an outerpath has queue number 1?
\end{open}

Our hardness results rule out several of the most natural parameters for parametrized algorithms for determining the queue number of a graph. In particular, parameterizing by treewidth, pathwidth or queue number cannot lead to FPT or XP algorithms unless $\mathrm{P}=\mathrm{NP}$. It remains unclear whether there are other meaningful graph parameters that allow such algorithms. As our reduction has two vertices of high degree, maximum vertex degree could be a candidate. Our existential results for outerpaths of queue number indicate that the maximum vertex degree is a relevant parameter for the queue number.

\begin{open}
Are there interesting graph parameters that allow FPT or XP algorithms for testing whether a graph has queue number~$1$ or for computing its queue number?
\end{open}

Recall that for outerplanar graphs, we have to decide only between queue number~$1$ or~$2$.

\bibliography{references,stacks,queues}
\end{document}